\documentclass[runningheads,a4paper]{llncs}

\usepackage[english]{babel}
\usepackage[ansinew]{inputenc}
\usepackage{amsmath}
\usepackage{amsfonts}
\usepackage{amssymb}
\usepackage{graphicx}
\usepackage{fancybox}
\usepackage{calc}
\usepackage{tabularx}
\usepackage{tikz}
\usepackage{array}
\usepackage{lmodern}
\usepackage{mathabx}
\usepackage{fancybox}
\usepackage{epstopdf}
\usepackage{bbm}
\usepackage{url}
\usepackage{hyperref}
\urldef{\mails}\path|{filipp.valovich, francesco.alda}@rub.de|
\newcommand\numberthis{\addtocounter{equation}{1}\tag{\theequation}}
\newcommand{\E}{\mathbb{E}}
\newcommand{\Mod}{\ \mathrm{mod}\ }

\DeclareMathOperator*{\argmax}{arg\,max}

\newtheorem{Thm}{Theorem}
\newtheorem{Lem}[Thm]{Lemma}

\newtheorem{Cor}{Corollary}

\newtheorem{Def}{Definition}

\newtheorem{Exm}{Example}

\newtheorem{Rem}{Remark}
\newtheorem{Not}{Note}
\newtheorem{Ntn}{Notation}
\newtheorem{Ctn}{Construction}

\title{Computational Differential Privacy\\ from Lattice-based Cryptography
\thanks{This is the full version of \cite{112}. The research was supported by the DFG Research Training Group GRK $1817/1$}
}
\titlerunning{CDP from Lattice-based Crypto}
\author{Filipp Valovich\and Francesco Ald\`{a}}
\institute{Horst G\"{o}rtz Institute for IT Security\\ Faculty of Mathematics\\ Ruhr-Universit\"{a}t Bochum, Universit\"{a}tsstrasse 150, 44801 Bochum, Germany\\ \mails}
\date{}

\allowdisplaybreaks

\begin{document}

\maketitle
\pagestyle{plain}

\begin{abstract} \noindent The emerging technologies for large scale data analysis raise new challenges to the security and privacy of sensitive user data. In this work we investigate the problem of private statistical analysis of time-series data in the distributed and semi-honest setting. In particular, we study some properties of Private Stream Aggregation (PSA), first introduced by Shi et al. $2011$. This is a computationally secure protocol for the collection and aggregation of data in a distributed network and has a very small communication cost. In the non-adaptive query model, a secure PSA scheme can be built upon any key-homomorphic \textit{weak} pseudo-random function as shown by Valovich $2017$, yielding security guarantees in the \textit{standard model} which is in contrast to Shi et. al. We show that every mechanism which preserves $(\epsilon,\delta)$-differential privacy in effect preserves \textit{computational} $(\epsilon,\delta)$-differential privacy when it is executed through a secure PSA scheme. Furthermore, we introduce a novel perturbation mechanism based on the \textit{symmetric Skellam distribution} that is suited for preserving differential privacy in the distributed setting, and find that its performances in terms of privacy and accuracy are comparable to those of previous solutions. On the other hand, we leverage its specific properties to construct a computationally efficient prospective post-quantum protocol for differentially private time-series data analysis in the distributed model. The security of this protocol is based on the hardness of a new variant of the Decisional Learning with Errors (DLWE) problem. In this variant the errors are taken from the symmetric Skellam distribution. We show that this new variant is hard based on the hardness of the standard Learning with Errors (LWE) problem where the errors are taken from the discrete Gaussian distribution. Thus, we provide a variant of the LWE problem that is hard based on conjecturally hard lattice problems and uses a discrete error distribution that is similar to the continuous Gaussian distribution in that it is closed under convolution. A consequent feature of the constructed prospective post-quantum protocol is the use of the same noise for security and for differential privacy.
\end{abstract}

\section{Introduction}

Among several challenges that the society is facing in the era of big data, the problem of data processing under strong privacy and security guarantees receives a lot of attention in research communities. The framework of statistical disclosure control aims at providing strong privacy guarantees for the records stored in a database while enabling accurate statistical analyses to be performed. In recent years, \textit{differential privacy} has become one of the most important paradigms for privacy-preserving statistical analyses. According to K. Nissim, a pioneer in this area of research, "there is a great promise for the marriage of Big Data and Differential Privacy".\footnote{\url{http://bigdata.csail.mit.edu/Big_Data_Privacy}} It combines mathematically rigorous privacy guarantees with highly accurate analyses over larger data sets. Generally, the notion of differential privacy is considered in the centralised setting where we assume the existence of a \textit{trusted curator} (see Blum et al. \cite{16}, Dwork \cite{24}, Dwork et al. \cite{8}, McSherry and Talwar \cite{9}) who collects data in the clear, aggregates and perturbs it properly (e.g. by adding Laplace noise) and publishes it. In this way, the output statistics are not significantly influenced by the presence (resp. absence) of a particular record in the database.\\
In this work we study how to preserve differential privacy when we cannot rely on a trusted curator. In this so-called \textit{distributed setting}, the users have to send their own data to an untrusted aggregator. Preserving differential privacy and achieving high accuracy in the distributed setting is of course harder than in the centralised setting, since the users have to execute a perturbation mechanism on their own. In order to achieve the same accuracy as provided by well-known techniques in the centralised setting, Shi et al. \cite{2} introduce the \textit{Private Stream Aggregation} (PSA) scheme, a cryptographic protocol enabling each user to securely send encrypted time-series data to an aggregator. The aggregator is then able to decrypt the aggregate of all data in each time-step, but cannot retrieve any further information about the individual data. Using such a protocol, the task of perturbation can be split among the users, such that \textit{computational} differential privacy, a notion first introduced by Mironov et al. \cite{15}, is preserved \textit{and} high accuracy is guaranteed. For a survey of applications of this protocol, we refer to \cite{2}.\par\smallskip

\noindent\textbf{Related Work.} In \cite{2}, a PSA scheme for sum queries was provided that satisfies strong security guarantees under the Decisional Diffie-Hell\-man (DDH) assumption. However, this instantiation has some limitations. First, the security only holds in the random oracle model; second, its decryption algorithm requires the solution of the discrete logarithm in a given range, which can be very time-consuming if the number of users and the plaintext space are large. Third, a connection between the security of a PSA scheme and computational differential privacy is not explicitly shown. In a subsequent work by Chan et al. \cite{3}, this connection is still not completely established, since the polynomial-time reduction between an attacker against a secure PSA scheme and a database distinguisher is missing.\\
By lowering the requirements of Aggregator Obliviousness introduced in \cite{2} by abrogating the attacker's possibility to \textit{adaptively compromise} users during the execution of a PSA scheme with time-series data, Valovich \cite{111} shows that a PSA scheme achieving this lower security level can be built upon any \textit{key-homomorphic weak pseudo-random function}. Since weak pseudo-randomness can be achieved in the standard model, this condition also enables secure schemes in the standard model. Furthermore, an instantiation of this result based on the DDH assumption was given in \cite{111}, where decryption is always efficient. Joye and Libert \cite{38} provide a protocol with the same security guarantees in the random oracle model as in \cite{2}. The security of their scheme relies on the Decisional Composite Residuosity assumption (rather than DDH as in \cite{2}) and as a result, in the security reduction they can remove a factor which is cubic in the number of users. However, their scheme involves a semi-trusted party for setting some public parameters. In this work, we provide an instantiation of the \textit{generic} PSA construction from \cite{111} which relies on the Decisional Learning with Errors (DLWE) assumption. While in this generic security reduction a \textit{linear} factor in the number of users cannot be avoided, our construction \textit{does not} involve any trusted party and has security guarantees in the standard model. In a subsequent work \cite{38seq}, a generalisation of the scheme from \cite{38} is obtained based on smooth projective hash functions (see \cite{38add}). This generalisation allows the construction of secure protocols based on various hardness assumptions. However, the dependencies on a semi-trusted party (for most of the instantiations) and on a random oracle remain. 

\noindent\textbf{Contributions.} In this regard, our results are as follows. First, reduction-based security proofs for cryptographic schemes usually require an attacker in the corresponding security game to send two different plaintexts (or plaintext collections) to a challenger. The adversary receives then back a ciphertext which is the encryption of one of these collections and has to guess which one it is. In any security definition for a PSA scheme, these collections must satisfy a particular requirement, i.e. they must lead to the same aggregate, since the attacker has the capability to decrypt the aggregate (different aggregates would make the adversary's task trivial). In general, however, this requirement cannot be satisfied in the context of differential privacy. Introducing a novel kind of security reduction which deploys a \textit{biased coin} flip, we show that, whenever a randomised perturbation procedure is involved in a PSA scheme, the requirement of having collections with equal aggregate can be abolished. This result can be generalised to any cryptographic scheme with such a requirement. Using this property, we are able to show that if a mechanism preserves differential privacy, then it preserves computational differential privacy when it is used as a randomised perturbation procedure in a PSA scheme. This provides the missing step in the analysis from \cite{2}.\\

Second, we introduce the \textit{Skellam mechanism} that uses the symmetric Skellam distribution and compare it with the geometric mechanism by Ghosh et al. \cite{52} and the binomial mechanism by Dwork et al. \cite{14}. All three mechanisms preserve differential privacy and make use of discrete probability distributions. Therefore, they are well-suited for an execution through a PSA scheme. For generating the right amount of noise among all users, these mechanisms apply two different approaches. While in the geometric mechanism, with high probability, only one user generates the noise necessary for differential privacy, the binomial and Skellam mechanisms allow all users to generate noise of small variance, that sums up to the required value for privacy by the reproducibility property of the binomial and the Skellam distributions. We show that the theoretical error bound of the Skellam mechanism is comparable to the other two. At the same time, we provide experimental results showing that the geometric and Skellam mechanisms have a comparable accuracy in practice, while beating the one of the binomial mechanism. The advantage of the Skellam mechanism is that, based on the previously mentioned results, it can be used it to construct the first secure, prospective post-quantum PSA scheme for sum queries that automatically preserves computational differential privacy. The corresponding weak pseudo-random function for this protocol is constructed from the \textit{Learning with Errors} (LWE) problem that received a lot of attention in the cryptographic research community in recent years. As an instance of the LWE problem we are given a uniformly distributed matrix $\textbf{\upshape A}\in\mathbb{Z}_q^{\lambda\times\kappa}$ and a noisy codeword $\textbf{\upshape y}=\textbf{\upshape Ax}+\textbf{\upshape e}\in\mathbb{Z}_q^{\lambda}$ with an error term $\textbf{\upshape e}\in\mathbb{Z}_q^{\lambda}$ chosen from a proper error distribution $\chi^{\lambda}$ and an unknown $\textbf{\upshape x}\in\mathbb{Z}_q^{\kappa}$. The task is to find the correct vector $\textbf{\upshape x}$. In the decisional version of this problem (DLWE problem) we are given $(\textbf{\upshape A}, \textbf{\upshape y})$ and have to decide whether $\textbf{\upshape y}=\textbf{\upshape Ax}+\textbf{\upshape e}$ or $\textbf{\upshape y}$ is a uniformly distributed vector in $\mathbb{Z}_q^{\lambda}$. Regev \cite{42} provided a search-to-decision reduction to show that the two problems are essentially equivalent in the worst case and Micciancio and Mol \cite{43} provided a sample preserving search-to-decision reduction for certain cases showing the equivalence in the average case. Moreover, in \cite{42} the average-case-hardness of the search problem was established by the construction of an efficient quantum algorithm for worst-case lattice problems using an efficient solver of the LWE problem if the error distribution $\chi$ is a \textit{discrete Gaussian} distribution. Accordingly, most cryptographic applications of the LWE problem used a discrete Gaussian error distribution for their constructions. We will take advantage of the reproducibility of the Skellam distribution for our DLWE-based PSA scheme by using errors following the symmetric Skellam distribution rather than the discrete Gaussian distribution, which is not reproducible. The result is that the sum of the errors generated by every user to secure their data is also a Skellam variable and therefore sufficient for preserving differential privacy. Hence, we show the average-case-hardness of the LWE problem with errors drawn from the Skellam distribution. Our proof is inspired by techniques used by D\"{o}ttling and M\"{u}ller-Quade \cite{39} where a variant of the LWE problem with uniform errors on a small support is shown to be hard.\footnote{Although the uniform distribution is reproducible as well, the result from \cite{39} does not provide a proper error distribution for our DLWE-based PSA scheme, since a differentially private mechanism with uniform noise provides no accuracy to statistical data analyses.} 
Consequently, we obtain a lattice-based secure PSA scheme for analysing sum queries under differential privacy where the noise is used both for security and for preserving differential privacy at once.\par\smallskip

\noindent\textbf{Other Related Work.} Another series of works deals with a distributed generation of noise for preserving differential privacy. Dwork et al. \cite{14} consider the Gaussian distribution for splitting the task of noise generation among all users. Their proposed scheme requires more interactions between the users than our solution. \'{A}cs and Castelluccia \cite{36} apply privacy-preserving data aggregation to smart metering. The generation of Laplace noise is performed in a distributed manner, since each meter simply generates the difference of two Gamma distributed random variables as a share of a Laplace distributed random variable. In the work \cite{19} each user generates a share of Laplace noise by generating a vector of four Gaussian random variables. For a survey of the mechanisms given in \cite{19} and \cite{36}, we refer to \cite{21}. However, the aforementioned mechanisms generate noise drawn according to continuous distributions, but for the use in a PSA scheme discrete noise is required. Therefore, we consider proper discrete distributions and compare their performances in private statistical analyses.

\section{Preliminaries}\label{prelsec}

\begin{Ntn} For a natural number $n$, we denote by $[n]$ the interval $\{1,\ldots,n\}$.
\end{Ntn}

\begin{Ntn} Let $X$ be a set. If $X$ is finite, we denote by $\mathcal{U}(X)$ the uniform distribution on $X$. Let $\chi$ be a distribution on $X$. We denote by $x\leftarrow\chi$ (or sometimes $x\leftarrow\chi(X)$) the sampling of $x$ from $X$ according to $\chi$. If $X$ has only two elements, we write $x\leftarrow_R X$ instead of $x\leftarrow\mathcal{U}(X)$. If $A\leftarrow\chi^{a\times b}$ (or $A\leftarrow\chi(X^{a\times b})$) then $A$ is an $a\times b$-matrix constructed by picking every entry independently from $X$ according to the distribution $\chi$.
\end{Ntn}

\begin{Ntn} Let $\kappa$ be a security parameter. If $\omega=\omega(\kappa)<1/\text{\upshape\sffamily poly}(\kappa)$ for every polynomial $\text{\upshape\sffamily poly}$ and all $\kappa>\kappa^\prime$, for some $\kappa^\prime\in\mathbb{N}$, then we say that $\omega$ is negligible in $\kappa$ and denote it by $\omega=\text{\upshape\sffamily neg}(\kappa)$. If $\omega$ is non-negligible in $\kappa$ (i.e. if $\omega> 1/\text{\upshape\sffamily poly}(\kappa)$), then we write $\omega>\text{\upshape\sffamily neg}(\kappa)$.
\end{Ntn}

\begin{Ntn} Let $q>2$ be a prime. We handle elements from $\mathbb{Z}_q$ as their central residue-class representation. This means that $x^\prime\in\mathbb{Z}_q$ is identified with $x\equiv x^\prime \Mod q$ for $x\in\{-(q-1)/2,\ldots,(q-1)/2\}$ thereby lifting $x^\prime$ from $\mathbb{Z}_q$ to $\mathbb{Z}$.
\end{Ntn}

\subsection{Problem statement}\label{mechov}

In this work we consider a distributed and semi-honest setting where $n$ users are asked to participate in some statistical analyses but do not trust the data analyst (or aggregator), who is assumed to be honest but curious. Therefore, the users cannot provide their own data in the clear. Moreover, they communicate solely and independently with the untrusted aggregator, who wants to analyse the users data by means of time-series queries and aims at obtaining answers as accurate as possible. More specifically, assume that the data items belong to a data universe $\mathcal{D}$. For a sequence of time-steps $t\in T$, where $T$ is a discrete time period, the analyst sends queries which are answered by the users in a distributed manner. Each query is modelled as a function $f:\mathcal{D}^n\to \mathcal{O}$ for a finite or countably infinite set of possible outputs (i.e. answers to the query) $\mathcal{O}$.\\ 
We also assume that some users may act in order to compromise the privacy of the other participants. More precisely, we assume the existence of a publicly known constant $\gamma\in(0,1]$ which is the a priori estimate of the lower bound on the fraction 
of uncompromised users who honestly follow the protocol and want to release useful information about their data (with respect to a particular query $f$), while preserving $(\epsilon,\delta)$-differential privacy. The remaining $(1-\gamma)$-fraction of users is assumed to be compromised and following the protocol but aiming at violating the privacy of uncompromised users. For that purpose, these users form a coalition with the analyst and send her auxiliary information, e.g. their own data in the clear.\\ 
For computing the answers to the aggregator's queries, a special cryptographic protocol, called Private Stream Aggregation (PSA) scheme, is used by \textit{all} users. In contrast to common secure multi-party techniques (see \cite{25}, \cite{20}), this protocol requires each user to send only one message per query to the analyst. In connection with a differentially private mechanism, a PSA scheme assures that the analyst is only able to learn a noisy aggregate of users' data (as close as possible to the real answer) and nothing else. Specifically, for preserving $(\epsilon,\delta)$-differential privacy, it would be sufficient to add a single copy of (properly distributed) noise $Y$ to the aggregated statistics. Since we cannot add such noise once the aggregate has been computed, the users have to generate and add noise to their original data in such a way that the sum of the errors has the same distribution as $Y$. For this purpose, we see two different approaches. In the first one, with small probability a user adds noise sufficient to preserve the privacy of the entire statistics. This probability is calibrated in such a way only one of the $n$ users is actually expected to add noise at all. Shi et al. \cite{2} investigate this method using the geometric mechanism from. \cite{52}. In the second approach, each user generates noise of small variance, such that the sum of all noisy terms suffices to preserve differential privacy of the aggregate. To achieve this goal, we need a discrete probability distribution which is closed under convolution and is known to provide differential privacy. The binomial mechanism from \cite{14} and the Skellam mechanism introduced in this work serve these purposes.\footnote{Due to the use of a cryptographic protocol, the plaintexts have to be discrete. This is the reason why we use discrete distributions for generating noise.}\\
Since the protocol used for the data transmission is computationally secure, the entire mechanism preserves a computational version of differential privacy as it will be shown in Section \ref{psa}.

\subsection{Definitions}

\subsubsection{Differential Privacy.}

We consider a database as an element $D\in\mathcal{D}^n$ with data universe $\mathcal{D}$ and number of users $n$. Since $D$ may contain sensitive information, the users want to protect their privacy. Therefore, a privacy-preserving mechanism must be applied. We will always assume that a mechanism is applied in the distributed setting. Differential privacy is a well-established notion for privacy-preserving statistical analyses. We recall that a randomised mechanism preserves differential privacy if its application on two adjacent data\-bases (databases differing in one entry only) leads to close distributions of the outputs.

\begin{Def}[Differential Privacy \cite{8}]
Let $\mathcal{R}$ be a (possibly infinite) set and let $n\in\mathbb{N}$. A randomised mechanism $\mathcal{A}:\mathcal{D}^n\to\mathcal{R}$ preserves $(\epsilon,\delta)$-differential privacy (short: \mbox{\upshape\sffamily DP}), if for all adjacent databases $D_0, D_1\in\mathcal{D}^n$ and all measurable $R\subseteq\mathcal{R}$:
\[\Pr[\mathcal{A}(D_0)\in R]\leq e^\epsilon\cdot \Pr[\mathcal{A}(D_1)\in R]+\delta.\]
The probability space is defined over the randomness of $\mathcal{A}$.
\end{Def}

The additional parameter $\delta$ is necessary for mechanisms which cannot preserve $\epsilon$-\mbox{\upshape\sffamily DP} (i.e. $(\epsilon,0)$-\mbox{\upshape\sffamily DP}) for certain cases. However, if the probability that these cases occur is bounded by $\delta$, then the mechanism preserves $(\epsilon,\delta)$-\mbox{\upshape\sffamily DP}.

In the literature, there are well-established mechanisms for preserving differential privacy, e.g. the \textit{Laplace mechanism} from \cite{8} and the \textit{Exponential mechanism} from \cite{9}. In order to privately evaluate a query, these mechanisms draw error terms according to some distribution depending on the query's global sensitivity.

\begin{Def}[Global Sensitivity]
The global sensitivity $S(f)$ of a query $f:\mathcal{D}^n\to\mathbb{R}^k$ is defined as
\[S(f)=\max_{D_0,D_1 \mbox{ \scriptsize\upshape\, adjacent}}||f(D_0)-f(D_1)||_1.\]
\end{Def}

In particular, we will consider sum queries $f_{\mathcal{D}}:\mathcal{D}^n\to\mathbb{Z}$ defined as $f_{\mathcal{D}}(D):=\sum_{i=1}^n d_i$, for $D=(d_1,\ldots,d_n)\in\mathcal{D}^n$ and $\mathcal{D}\subseteq\mathbb{Z}$.\\
For measuring how well the output of a mechanism estimates the real data with respect to a particular query, we use the notion of $(\alpha,\beta)$-accuracy.

\begin{Def}[Accuracy]
The output of a mechanism $\mathcal{A}$ achieves $(\alpha,\beta)$-accuracy for a query $f:\mathcal{D}^n\to\mathbb{R}$ if for all $D\in\mathcal{D}^n$:
\[\Pr[|\mathcal{A}(D)-f(D)|\leq\alpha]\geq 1-\beta.\]
The probability space is defined over the randomness of $\mathcal{A}$.
\end{Def}

The use of a cryptographic protocol for transferring data 
provides a computational security level. If such a protocol is applied to preserve \mbox{\upshape\sffamily DP}, this implies that only a computational level of \mbox{\upshape\sffamily DP} can be provided. The definition of computational differential privacy was first provided in \cite{15} and subsequently extended in \cite{3}.

\begin{Def}[Computational Differential Privacy \cite{3}] Let $\kappa$ be a security parameter and $n\in\mathbb{N}$ with $n=\text{\upshape\sffamily poly}(\kappa)$. A randomised mechanism $\mathcal{A}:\mathcal{D}^n\to\mathcal{R}$ preserves computational $(\epsilon,\delta)$-differential privacy (short: \mbox{\upshape\sffamily CDP}), if for all adjacent databases $D_0, D_1\in\mathcal{D}^n$ and all probabilistic polynomial-time distinguishers $\mathcal{D}_{\mbox{\scriptsize\upshape\sffamily CDP}}$:
\[\Pr[\mathcal{D}_{\mbox{\scriptsize\upshape\sffamily CDP}}(1^\kappa,\mathcal{A}(D_0))=1]\leq e^\epsilon\cdot\Pr[\mathcal{D}_{\mbox{\scriptsize\upshape\sffamily CDP}}(1^\kappa,\mathcal{A}(D_1))=1]+\delta+\text{\upshape\sffamily neg}(\kappa),\]
where $\text{\upshape\sffamily neg}(\kappa)$ is a negligible function in $\kappa$. The probability space is defined over the randomness of $\mathcal{A}$ and $\mathcal{D}_{\mbox{\scriptsize\upshape\sffamily CDP}}$.
\end{Def}

The notion of \mbox{\upshape\sffamily CDP} is a natural computational indistinguishability-extension of the infor\-mation-theoretical definition. The advantage is that preserving differential privacy only against bounded attackers helps to substantially reduce the error of the answer provided by the mechanism. 

\subsubsection{Private Stream Aggregation.}
 
We define the Private Stream Aggregation scheme and give a security definition for it. Thereby, we mostly follow the concepts introduced in \cite{2}, though we deviate in a few points. A PSA scheme is a protocol for safe distributed time-series data transfer which enables the receiver (here: the untrusted analyst) to learn nothing else than the sums $\sum_{i=1}^n x_{i,j}$ for $j=1,2,\ldots$, where $x_{i,j}$ is the value of the $i$th participant in time-step $j$ and $n$ is the number of participants (or users). Such a scheme needs a key exchange protocol for all $n$ users together with the analyst as a precomputation (e.g. using multi-party techniques), and requires each user to send exactly one message in each time-step $j=1,2,\ldots$.

\begin{Def}[Private Stream Aggregation \cite{2}]
Let $\kappa$ be a security parameter, $\mathcal{D}$ a set and $n=\text{poly}(\kappa)$, $\lambda=\text{poly}(\kappa)$. A \text{\upshape Private Stream Aggregation} (PSA) scheme $\Sigma=(\mbox{\upshape\sffamily Setup}, \mbox{\upshape\sffamily PSAEnc}, \mbox{\upshape\sffamily PSADec})$ is defined by three ppt algorithms:
\begin{description}
\item \textbf{\mbox{\upshape \sffamily Setup}}: $(\mbox{\upshape\sffamily pp},T,s_0,s_1,\ldots,s_n)\leftarrow \mbox{\upshape\sffamily Setup}(1^\kappa)$ with public parameters $\mbox{\upshape\sffamily pp}$,\linebreak $T=\{t_1,\ldots,t_\lambda\}$ and secret keys $s_i$ for all $i=1,\ldots,n$.
\item \textbf{\mbox{\upshape \sffamily PSAEnc}}: For $t_j\in T$ and all $i=1,\ldots,n$: $c_{i,j}\leftarrow \mbox{\upshape\sffamily PSAEnc}_{s_i}(t_j,x_{i,j})\mbox{ for } x_{i,j}\in\mathcal{D}$.
\item \textbf{\mbox{\upshape \sffamily PSADec}}: Compute $\sum_{i=1}^n x'_{i,j}=\mbox{\upshape\sffamily PSADec}_{s_0}(t_j,c_{1,j},\ldots,c_{n,j})$ for $t_j\in T$ and ciphers $c_{1,j},\ldots,c_{n,j}$. For all $t_j\in T$ and $x_{1,j},\ldots,x_{n,j}\in\mathcal{D}$ the following holds:
\[\mbox{\upshape\sffamily PSADec}_{s_0}(t_j, \mbox{\upshape\sffamily PSAEnc}_{s_1}(t_j,x_{1,j}),\ldots,\mbox{\upshape\sffamily PSAEnc}_{s_n}(t_j,x_{n,j}))=\sum_{i=1}^n x_{i,j}.\]
\end{description}
\end{Def}

The Setup-phase has to be carried out just once and for all, and can be performed with a secure multi-party protocol among all users and the analyst. In all other phases, no communication between the users is needed.\\
The system parameters $\mbox{\upshape\sffamily pp}$ are public and constant for all time-steps with the implicit understanding that they are used in $\Sigma$. Every user encrypts her value $x_{i,j}$ with her own secret key $s_i$ and sends the ciphertext to the analyst. If the analyst receives the ciphertexts of \textit{all} users in a time-step $t_j$, it computes the aggregate with the decryption key $s_0$.\\
For a particular time-step, let the users' values be of the form $x_{i,j}=d_{i,j}+e_{i,j}$, $i=1,\ldots,n$, where $d_{i,j}\in\mathcal{D}$ is the original data of the user $i$ and $e_{i,j}$ is her error term.\footnote{The perturbation of data is considered in the context of \mbox{\upshape\sffamily DP}. Potentially, it yields larger values $x_{i,j}$ due to the (possibly) infinite domain of the underlying probability distribution. Depending on the variance, we therefore need to choose a sufficiently large interval $\widehat{\mathcal{D}}=\{-m,\ldots,m\}$ as plaintext space, where $m>\max\{|w|,|w^\prime|\}$ such that $|x_{i,j}|\leq m$ for all $i=1,\ldots,n$ with high probability.} It is reasonable to assume that $e_{i,j}=0$ for the $(1-\gamma)\cdot n$ compromised users, since this can only increase their chances to infer some information about the uncompromised users. There is no privacy-breach if only one user adds the entirely needed noise (first approach) or if the uncompromised users generate noise of low variance (second approach), since the single values $x_{i,j}$ are encrypted and the analyst cannot learn anything about them, except for their aggregate.\\
 
\noindent\textit{Security.} Since our model allows the analyst to compromise users, the aggregator can obtain auxiliary information about the data of the compromised users or their secret keys. Even then a secure PSA scheme should release no more information than the aggregate of the uncompromised users' data.\\
Informally, a PSA scheme $\Sigma$ is secure if every probabilistic polynomial-time algorithm, with knowledge of the analyst's and compromised users' keys and with adaptive encryption queries, has only negligible advantage in distinguishing between the encryptions of two databases $\widehat{D}_0, \widehat{D}_1$ of its choice with equal aggregates. We can assume that an adversary knows the secret keys of the entire compromised coalition. 
If the protocol is secure against such an attacker, then it is also secure against an attacker without the knowledge of every key from the coalition. Thus, in our security definition we consider the most powerful adversary.

\begin{Def}[Non-adaptive Aggregator Obliviousness \cite{111}]\label{securitygame}
Let $\kappa$ be a security parameter. Let $\mathcal{T}$ be a ppt adversary for a PSA scheme\linebreak $\Sigma=$ $(\mbox{\upshape\sffamily Setup}, \mbox{\upshape\sffamily PSAEnc}, \mbox{\upshape\sffamily PSADec})$ and let $\mathcal{D}$ be a set. We define a security game between a challenger and the adversary $\mathcal{T}$.
\begin{description}
 \item\textbf{Setup.} The challenger runs the \text{\upshape\sffamily Setup} algorithm on input security parameter $\kappa$ and returns public parameters $\mbox{\upshape\sffamily pp}$, public encryption parameters $T$ with $|T|=\lambda=\text{poly}(\kappa)$ and secret keys $s_0,s_1,\ldots,s_n$. It sends $\kappa,\mbox{\upshape\sffamily pp}, T, s_0$ to $\mathcal{T}$. $\mathcal{T}$ chooses $U\subseteq[n]$ and sends it to the challenger which returns $(s_i)_{i\in[n]\setminus U}$.
\item\textbf{Queries.} $\mathcal{T}$ is allowed to query $(i,t_j,x_{i,j})$ with $i\in U, t_j\in T, x_{i,j}\in\mathcal{D}$ and the challenger returns $c_{i,j}\leftarrow\mbox{\upshape\sffamily PSAEnc}_{s_i}(t_j,x_{i,j})$.
\item\textbf{Challenge.} $\mathcal{T}$ chooses $t_{j^*}\in T$ such that no encryption query with $t_{j^*}$ was made. (If there is no such $t_{j^*}$ then the challenger simply aborts.) $\mathcal{T}$ queries two different tuples $(x_{i,j^*}^{[0]})_{i\in U},(x_{i,j^*}^{[1]})_{i\in U}$ with
$\sum_{i\in U} x_{i,j^*}^{[0]}=\sum_{i\in U} x_{i,j^*}^{[1]}$.
The challenger flips a random bit $b\leftarrow_R\{0,1\}$. For all $i\in U$ the challenger returns $c_{i,j^*}\leftarrow\mbox{\upshape\sffamily PSAEnc}_{s_i}(t_{j^*},x_{i,j^*}^{[b]})$.
\item\textbf{Queries.} $\mathcal{T}$ is allowed to make the same type of queries as before restricted to encryption queries with $t_j\neq t_{j^*}$.
\item\textbf{Guess.} $\mathcal{T}$ outputs a guess about $b$.
\end{description}
The adversary's probability to win the game (i.e. to guess $b$ correctly) is $1/2+\nu(\kappa)$. A PSA scheme is \text{\upshape non-adaptively} \text{\upshape aggregator} \text{\upshape oblivious} or achieves \text{\upshape non-adaptive} \text{\upshape Aggregator}\linebreak \text{\upshape Obliviousness} ($\mbox{\upshape\sffamily AO}2$), if there is no ppt adversary $\mathcal{T}$ with advantage $\nu(\kappa)>\text{\upshape\sffamily neg}(\kappa)$ in winning the game.
\end{Def}

Encryption queries are made only for $i\in U$, since knowing the secret key for all $i\in[n]\setminus U$ the adversary can encrypt a value autonomously. If encryption queries in time-step $t_j^*$ were allowed, then no deterministic scheme would be aggregator oblivious. The adversary $\mathcal{T}$ can determine the original data of all $i\in[n]\setminus U$ for every time-step, since it knows $(s_i)_{i\in[n]\setminus U}$. Then $\mathcal{T}$ can compute the aggregate of the uncompromised users' data.\\ 
Definition \ref{securitygame} differs from the definition of Aggregator Obliviousness from \cite{2} in that we require the adversary to specify the set $U$ of uncompromised users \textit{before} making any query, i.e. we do not allow the adversary to determine $U$ adaptively. In light of that and in analogy to the definitions of security against adaptive ($\mbox{\upshape\sffamily CCA}1$) and non-adaptive ($\mbox{\upshape\sffamily CCA}2$) chosen ciphertext adversaries, we refer to the security definition from \cite{2} as $\mbox{\upshape\sffamily AO}1$ and to our security definition as $\mbox{\upshape\sffamily AO}2$. In this work we simply call a PSA scheme \textit{secure} if it achieves $\mbox{\upshape\sffamily AO}2$.

\subsubsection{Weak PRF.}

In the analysis of the secure protocol, we make use of the following definition.

\begin{Def}[Weak PRF \cite{26}] Let $\kappa$ be a security parameter. Let $A,B,C$ be sets with sizes parameterised by a complexity parameter $\kappa$. A family of functions 
\[\mathcal{F}=\{\text{\upshape\sffamily F}_a\,|\,\text{\upshape\sffamily F}_a: B\to C\}_{a\in A}\] 
is called a \text{\upshape weak PRF family}, if for all ppt algorithms $\mathcal{D}_{\mbox{\scriptsize PRF}}^{\mathcal{O}(\cdot)}$ with oracle access to $\mathcal{O}(\cdot)$ (where $\mathcal{O}(\cdot)\in\{\text{\upshape\sffamily F}_a(\cdot),\text{\upshape\sffamily rand}(\cdot)\}$) on any polynomial number of given uniformly chosen inputs, we have:
\[|\Pr[\mathcal{D}_{\mbox{\scriptsize PRF}}^{\text{\upshape\sffamily F}_a(\cdot)}(\kappa)=1]-\Pr[\mathcal{D}_{\mbox{\scriptsize PRF}}^{\text{\upshape\sffamily rand}(\cdot)}(\kappa)=1]|\leq\text{\upshape\sffamily neg}(\kappa),\]
where $a\leftarrow\mathcal{U}(A)$ and $\text{\upshape\sffamily rand}\in\{f\,|\,f:B\to C\}$ is a random mapping from $B$ to $C$.
\end{Def}

\section{Main Result}

In this work we prove the following result by showing the connection between a key-homomorphic weak pseudo-random function and a differentially private mechanism for sum queries.

\begin{Thm}\label{mainthm} Let $\epsilon>0$, $w<w^\prime\in\mathbb{Z}$, $m,n\in\mathbb{N}$ with $\max\{|w|,|w^\prime|\}<m$. Let\linebreak $\mathcal{D}=\{w,\ldots,w^\prime\}$ and $f_{\mathcal{D}}$ be a sum query. If there exist groups $G^\prime\subseteq G$, a key-homomorphic weak pseudo-random function family mapping into $G^\prime$ and an efficiently computable and efficiently invertible homomorphism $\varphi:\{-mn,\ldots,mn\}\to G$ injective over $\{-mn,\ldots,mn\}$, then there exists an efficient mechanism for $f_{\mathcal{D}}$ that preserves $(\epsilon,\delta)$-\mbox{\upshape\sffamily CDP} for any $0<\delta<1$ with an error bound of $O(S(f_{\mathcal{D}})/\epsilon\cdot\log(1/\delta))$ and requires each user to send exactly one message.
\end{Thm}

The proof of Theorem \ref{mainthm} is provided in the next two sections. In Section \ref{psa} we recall from \cite{111} how to construct a general PSA scheme from a key-homomorphic weak PRF. Subsequently, we show that a secure PSA scheme in composition with a \mbox{\upshape\sffamily DP}-mechanism preserves \mbox{\upshape\sffamily CDP}. In Section \ref{lwepsasec}, based on the DLWE problem with errors drawn from a Skellam distribution, we provide an instantiation of a key-homomorphic weak PRF. This yields a concrete efficient PSA scheme that automatically embeds a \mbox{\upshape\sffamily DP}-mechanism with accuracy as stated in Theorem \ref{mainthm}.

\section{From Key-Homomorphic Weak PRF to \mbox{\upshape\sffamily CDP}}\label{psa}

We give a condition for the existence of secure PSA schemes and then analyse its connection to \mbox{\upshape\sffamily CDP}. 

\subsection{From Key-Homomorphic Weak PRF to secure PSA}

Now we state the condition for the existence of secure PSA schemes for sum queries in the sense of Definition \ref{securitygame}.

\begin{Thm}[Weak PRF gives secure PSA scheme \cite{111}]\label{PSATHEOREM}
Let $\kappa$ be a security parameter, and $m,n\in\mathbb{N}$ with $\log(m)=\text{\upshape\sffamily poly}(\kappa),n=\text{\upshape\sffamily poly}(\kappa)$. Let $(G,\cdot), (S,*)$ be finite abelian groups and $G^\prime\subseteq G$. For some finite set $M$, let $\mathcal{F}=\{\text{\upshape\sffamily F}_s\,|\,\text{\upshape\sffamily F}_s:M\to G^\prime\}_{s\in S}$ be a (possibly randomised) weak PRF family and let\linebreak $\varphi:\{-mn,\ldots,mn\}\to G$ be a mapping. Then the following PSA scheme $\Sigma=(\mbox{\upshape\sffamily Setup}, \mbox{\upshape\sffamily PSAEnc}, \mbox{\upshape\sffamily PSADec})$ achieves $\mbox{\upshape\sffamily AO}2$:
\begin{description}
\item \textbf{\mbox{\upshape \sffamily Setup}}: $(\mbox{\upshape\sffamily pp},T,s_0,s_1,\ldots,s_n)\leftarrow \mbox{\upshape\sffamily Setup}(1^\kappa)$, where $\mbox{\upshape\sffamily pp}$ are parameters of $G,G^\prime,S,$ $M,\mathcal{F},\varphi$. The keys are $s_i\leftarrow\mathcal{U}(S)$ for all $i\in[n]$ with $s_0=(\bigast_{i=1}^n s_i)^{-1}$ and $T\subset M$ such that all $t_j\in T$ are chosen uniformly at random from $M$, $j=1,\ldots,\lambda=\text{\upshape\sffamily poly}(\kappa)$.
\item \textbf{\mbox{\upshape \sffamily PSAEnc}}: Compute $c_{i,j}=\text{\upshape\sffamily F}_{s_i}(t_j)\cdot \varphi(x_{i,j})$ in $G$ for $x_{i,j}\in\widehat{\mathcal{D}}=\{-m,\ldots,m\}$ and public parameter $t_j\in T$.
\item \textbf{\mbox{\upshape \sffamily PSADec}}: Compute $V_j=\varphi^{-1}(S_j)$ (if possible) with $S_j=\text{\upshape\sffamily F}_{s_0}(t_j)\cdot c_{1,j}\cdot\ldots\cdot c_{n,j}$.
\end{description}
Moreover, if $\mathcal{F}$ contains only deterministic functions that are homomorphic over $S$, if $\varphi$ is homomorphic and injective over $\{-mn,\ldots,mn\}$ and if the $c_{i,j}$ are encryptions of the $x_{i,j}$, then $V_j=\sum_{i=1}^n x_{i,j}$, i.e. then \mbox{\upshape \sffamily PSADec} correctly decrypts $\sum_{i=1}^n x_{i,j}$.
\end{Thm}

The reason for not including the correctness property in the main statement is that in Section \ref{lwepsasec} we will provide an example of a secure PSA scheme based on the DLWE problem that does not have a fully correct decryption algorithm, but a noisy one. This noise is necessary for establishing the security of the protocol and will be also used for preserving the differential privacy of the decryption output.\\
Hence, we need a key-homomorphic weak PRF and a mapping which homomorphically aggregates all users' data. Since every data value is at most $m$, the scheme correctly retrieves the aggregate, which is at most $m\cdot n$. Importantly, the product of all pseudo-random values $\text{\upshape\sffamily F}_{s_0}(t),\text{\upshape\sffamily F}_{s_1}(t),\ldots,\text{\upshape\sffamily F}_{s_n}(t)$ is the neutral element in the group $G$ for all $t\in T$. Since the values in $T$ are uniformly distributed in $M$, it is enough to require that $\mathcal{F}$ is a \textit{weak} PRF family. Thus, the statement of Theorem \ref{PSATHEOREM} does not require a random oracle.

\subsection{From secure PSA to \mbox{\upshape\sffamily CDP}}\label{cdp}

In this section, we describe how to preserve \mbox{\upshape\sffamily CDP} using a PSA scheme. Specifically, let $\mathcal{A}$ be a mechanism which, given some event $Good$, evaluates a statistical query $f:\mathcal{D}^n\to\mathcal{O}$ over a database $D\in\mathcal{D}^n$ preserving $\epsilon$-\mbox{\upshape\sffamily DP}. Furthermore, let $\Sigma$ be a PSA scheme for $f$ that achieves $\mbox{\upshape\sffamily AO}1$ (or $\mbox{\upshape\sffamily AO}2$). We show that the composition of $\mathcal{A}$ and $\Sigma$ preserves $\epsilon$-\mbox{\upshape\sffamily CDP} given $Good$. Assume $\Pr[\neg Good]\leq\delta$. We show that $\mathcal{A}$ preserves $(\epsilon,\delta)$-\mbox{\upshape\sffamily CDP} \textit{unconditionally}, if executed through $\Sigma$. For simplicity, in this section we focus on sum-queries, but our analysis can be easily extended to more general statistical queries.\\
Since we are dealing with a computationally secure scheme, we need a polynomial time reduction from an adversary on the security of the PSA scheme to a \mbox{\upshape\sffamily CDP}-distinguisher for databases. As a result, our scheme preserves $(\epsilon,\delta)$-\mbox{\upshape\sffamily CDP}. Since in the original security-game, the generation of noise is not considered and the aggregates of the two challenge databases in the clear have to be equal with respect to the query $f_{\mathcal{D}}$ (rather than differ by $S(f_{\mathcal{D}})$ as in the definition of differential privacy), we need to perform a certain kind of randomisation process within the reduction. For this purpose, we introduce a novel kind of security game which we call a \textit{biased security game}. In contrast to a usual security game as it is common in security proofs for cryptosystems, in a biased security game the challenger is allowed to generate a biased coin in order to pick the data set to send as challenge to the adversary.

\subsubsection{Redefining the security of PSA}

Let us first modify the security game in Definition \ref{securitygame} in the following way. Let \textbf{game}~$\boldsymbol 1$ be the original game from Definition \ref{securitygame}. Let $p\in(0,1)$ and $P=\max\{p,1-p\}$. The \textbf{game}~$\boldsymbol 1_{\boldsymbol B}$\footnote{The subscript $\boldsymbol B$ stands for biased.} for a ppt adversary $\mathcal{T}_{1,B}$ is defined as \textbf{game}~$\boldsymbol 1$ with the following changes:
\begin{itemize}
 \item In the challenge phase, after choosing the two different challenge-collections, $\mathcal{T}_{1,B}$ computes $p\in(0,1)$ and sends it to the challenger.
 \item In the challenge phase, the challenger chooses $b=0$ with probability $p$ and $b=1$ with probability $1-p$.
\end{itemize}
We call a PSA scheme biased-secure, if the probability of every ppt adversary $\mathcal{T}_{1,B}$ in winning the above game is $P+\mbox{\upshape\sffamily neg}(\kappa)$. Note that \textbf{game}~$\boldsymbol 1$ is a special case of \textbf{game}~$\boldsymbol 1_{\boldsymbol B}$, where $P=1/2$. We refer to this case as the unbiased version (rather the biased version if $P>1/2$) of \textbf{game}~$\boldsymbol 1_{\boldsymbol B}$. In the unbiased case, we just drop the dependence on $P$ and the adversary is not required to send $p$ to its challenger.\\
The reason for introducing \textbf{game}~$\boldsymbol 1_{\boldsymbol B}$ is that in the context of differential privacy, it is very unlikely to have equal aggregates of the two plaintext collections in the challenge phase of the original game, since the data is perturbed by a differentially private mechanism. The use of a bias towards one of the collections will balance the incorporation of noise as we will show below. First, we show that a successful adversary in \textbf{game}~$\boldsymbol 1_{\boldsymbol B}$ for \textit{any} $P\in[1/2,1)$ yields a successful adversary in \textbf{game}~$\boldsymbol 1$.

\begin{Lem}\label{gamePoneone} Let $\kappa$ be a security parameter. For any $p\in(0,1)$, let $\mathcal{T}_{1,B}$ be an adversary in \textbf{\upshape game $\boldsymbol 1_{\boldsymbol B}$} with advantage $\nu_{1,B}(\kappa)>\text{\upshape\sffamily neg}(\kappa)$. Then there exists an adversary $\mathcal{T}_1$ in \textbf{\upshape game $\boldsymbol 1$} with advantage $\nu_{1}(\kappa)>\text{\upshape\sffamily neg}(\kappa)$.
\end{Lem}
\begin{proof} Given a successful adversary $\mathcal{T}_{1,B}$ in \textbf{game}~$\boldsymbol 1_{\boldsymbol B}$, we construct a successful adversary $\mathcal{T}_1$ in \textbf{game}~$\boldsymbol 1$ as follows:
\noindent\begin{description}
 \item\textbf{Setup.} Receive $\kappa, \mbox{\upshape\sffamily pp},T,s_0$ from the \textbf{game}~$\boldsymbol 1$-challenger and send it to $\mathcal{T}_{1,B}$. Receive $U\subseteq[n]$ from $\mathcal{T}_{1,B}$ and send it to the challenger. Forward the obtained response $(s_i)_{i\in[n]\setminus U}$ to $\mathcal{T}_{1,B}$.
\item\textbf{Queries.} Forward $\mathcal{T}_{1,B}$'s queries $(i,j,x_{i,j})$ with $i\in U, j\in[\lambda], x_{i,j}\in\widehat{\mathcal{D}}$ to the challenger, forward the response $c_{i,j}$ to $\mathcal{T}_{1,B}$.
\item\textbf{Challenge.} $\mathcal{T}_{1,B}$ chooses $j^*\in\{1,\ldots,\lambda\}$ such that no encryption query in $j^*$ was made, sends $p\in(0,1)$ and queries two different tuples $(x_{i,j^*}^{[0]})_{i\in U}$,\linebreak $(x_{i,j^*}^{[1]})_{i\in U}$ with $\sum_{i\in U} x_{i,j^*}^{[0]}=\sum_{i\in U} x_{i,j^*}^{[1]}$. Choose a bit $a$ with\linebreak $\Pr[a=0]=p, \Pr[a=1]=1-p$ and query $(x_{i,j^*}^{[a]})_{i\in U},(x_{i,j^*})_{i\in U}$ to the challenger, where $x_{i,j^*}\leftarrow\mathcal{U}(\widehat{\mathcal{D}})$ for all $i\in U$ and $\sum_{i\in U} x_{i,j^*}=\sum_{i\in U} x_{i,j^*}^{[a]}$. Obtain the response $(c_{i,j^*})_{i\in U}$ and forward it to $\mathcal{T}_{1,B}$. 
\item\textbf{Queries.} $\mathcal{T}_{1,B}$ can make the same type of queries as before with the restriction that no encryption query in $j^*$ can be made.
\item\textbf{Guess.} $\mathcal{T}_{1,B}$ gives a guess about $a$. If it is correct, output $0$; else, output $1$.
\end{description}
If $\mathcal{T}_{1,B}$ has output the correct guess about $a$, then $\mathcal{T}_1$ can say with high confidence that the challenge ciphertexts are the encryptions of $(x_{i,j^*}^{[a]})_{i\in U}$ and therefore outputs $0$. On the other hand, if $\mathcal{T}_{1,B}$'s guess was not correct, then $\mathcal{T}_1$ can say with high confidence that the challenge ciphertexts are the encryptions of the random collection $(x_{i,j^*})_{i\in U}$ and it outputs $1$. Formally:\par\medskip

\noindent\textbf{Case $\boldsymbol 1$.} Let $(c_{i,j^*})_{i\in U}=(\text{\upshape\sffamily PSAEnc}_{s_i}(t_{j^*},x_{i,j^*}^{[a]}))_{i\in U}$. Then $\mathcal{T}_1$ perfectly simulates \textbf{game}~$\boldsymbol 1_{\boldsymbol B}$ for $\mathcal{T}_{1,B}$ and the distribution of ciphers is the same as in \textbf{game}~$\boldsymbol 1_{\boldsymbol B}$:
\begin{align*} & \Pr[\mathcal{T}_1 \text{ wins}]\\
 = & \Pr[\mathcal{T}_1 \text{ outputs } 0]\\
 = & p\cdot \Pr[\mathcal{T}_{1,B} \text{ outputs } 0 \,|\, a=0]+(1-p)\cdot\Pr[\mathcal{T}_{1,B} \text{ outputs } 1 \,|\, a=1]\\
 = & \Pr[\mathcal{T}_{1,B} \text{ wins } \boldsymbol P\text{-\textbf{game} } \boldsymbol 1]\\
 = & P + \nu_{1,B}(\kappa).
\end{align*}\par\medskip
\noindent\textbf{Case $\boldsymbol 2$.} Let $(c_{i,j^*})_{i\in U}=(\text{\upshape\sffamily PSAEnc}_{s_i}(t_{j^*},x_{i,j^*}))_{i\in U}$. Then the ciphertexts are random with the constraint 
that their product is the same as in the first case. The probability that $\mathcal{T}_{1,B}$ wins his modified game is at most $P$ (by choosing the Bayes-optimal hypthesis) and
\begin{align*} & \Pr[\mathcal{T}_1 \text{ wins}]\\
 = & \Pr[\mathcal{T}_1 \text{ outputs } 1]\\
 = & \,p\cdot\Pr[\mathcal{T}_{1,B} \text{ outputs } 1 \,|\, a=0]+(1-p)\cdot\Pr[\mathcal{T}_{1,B} \text{ outputs } 0 \,|\, a=1]\cdot (1-p)\\
 = & \Pr[\mathcal{T}_{1,B} \text{ loses}]\\
 \geq &\, 1-P.
\end{align*}
Finally, we obtain that the advantage of $\mathcal{T}_1$ in winning \textbf{game}~$\boldsymbol 1$ is
\begin{align*}\nu_{1}(\kappa) & =|\Pr[\mathcal{T}_1 \text{ outputs } 1| (c_{i,j^*})_{i\in U}=(\text{\upshape\sffamily PSAEnc}_{s_i}(t_{j^*},x_{i,j^*}))_{i\in U}]\\
 & -\Pr[\mathcal{T}_1 \text{ outputs } 1| (c_{i,j^*})_{i\in U}=(\text{\upshape\sffamily PSAEnc}_{s_i}(t_{j^*},x_{i,j^*}^{[a]}))_{i\in U}]|\\
 & \geq\nu_{1,B}(\kappa)\\
 & >\text{\upshape\sffamily neg}(\kappa).
 \end{align*}
\hfill$\qed$\end{proof}

\subsubsection{Constructing a PSA adversary using a \mbox{\upshape\sffamily CDP} adversary}

\textit{Security game for adjacent databases.} For showing that a biased-secure PSA scheme is suitable for preserving \mbox{\upshape\sffamily CDP}, we have to construct a successful adversary in \textbf{game}~$\boldsymbol 1_{\boldsymbol B}$ using a successful distinguisher for adjacent databases. We define the following \textbf{game}~$\boldsymbol 0$ for a ppt adversary $\mathcal{T}_0$ which is identical to \textbf{game}~$\boldsymbol 1$ with a changed challenge-phase:
\begin{description}
\item\textbf{Challenge.} $\mathcal{T}_0$ chooses $j^*\in\{1,\ldots,\lambda\}$ with no encryption query made in $j^*$. $\mathcal{T}_0$ queries two \textit{adjacent} tuples $(d_{i,j^*}^{[0]})_{i\in U},(d_{i,j^*}^{[1]})_{i\in U}$. The challenger flips a random bit $b\leftarrow_R\{0,1\}$. For all $i\in U$, the challenger returns 
\[c_{i,j^*}=\mbox{\upshape\sffamily PSAEnc}_{s_i}(t_{j^*},x_{i,j^*}^{[b]}),\]
where $x_{i,j^*}^{[b]}\in\widehat{\mathcal{D}}$ is a noisy version of $d_{i,j^*}^{[b]}\in\mathcal{D}$ for all $i\in U$ obtained by some randomised perturbation process.
\end{description}
Now consider the following experiment which we call $\mbox{\upshape \sffamily Exp}_1$. Let $\chi:\widehat{\mathcal{D}}^n\to[0,1]$ be a probability mass function on $\widehat{\mathcal{D}}^n$. For simplicity, we consider the case where $\mathcal{D}\subseteq\widehat{\mathcal{D}}=\mathbb{Z}$. Let $D_0=(d_1^{(0)},\ldots,d_n^{(0)}),D_1=(d_1^{(1)},\ldots,d_n^{(1)})\in\mathcal{D}^n$ be adjacent, i.e. differing in only one entry. Then $\mbox{\upshape \sffamily Exp}_1$ is performed as follows:
\begin{enumerate} 
\item Let $B_{1/2}$ be a Bernoulli variable with $\Pr[B_{1/2}=0]=1/2$. Choose a realisation $b$ of $B_{1/2}$.
\item Let $\textbf{X}_1=(X_1,\ldots,X_n)$ be a random vector with 
\[\Pr[\textbf{X}_1=\textbf{x}]=\begin{cases} \chi(\textbf{x}-D_0) &\mbox{ if } b=0\\ \chi(\textbf{x}-D_1) &\mbox{ if } b=1,\end{cases}\]
i.e. $\textbf{X}_1=D_{B_{1/2}}+\textbf{E}$, where $\textbf{E}=(E_1,\ldots,E_n)$ is a random variable with $\Pr[\textbf{E}=\textbf{e}]=\chi(\textbf{e})$. Choose a realisation $\textbf{e}=(e_1,\ldots,e_n)$ of $\textbf{E}$ and let $\textbf{x}=(x_1,\ldots,x_n)$ with $x_i=d_i^{(b)}+e_i$ for all $i=1,\ldots,n$ be a realisation of $\textbf{X}_1$.
\item Compute $s=\sum_{i=1}^n x_i=\sum_{i=1}^n d_i^{(b)}+e_i$, which is a realisation of the random variable $S=\sum_{i=1}^n X_i=\sum_{i=1}^n d_i^{(B_{1/2})}+E_i$.
\item Output $b,\textbf{x},s$.
\end{enumerate}
We define an experiment $\mbox{\upshape \sffamily Exp}_2$.
\begin{enumerate} 
\item Choose a realisation $s$ of the random variable $S$ as defined in $\mbox{\upshape \sffamily Exp}_1$, i.e. choose a realisation $b$ of $B_{1/2}$, choose a realisation $\textbf{e}=(e_1,\ldots,e_n)$ of $\textbf{E}=(E_1,\ldots,E_n)$ according to the probability mass function $\chi$ and set $s=\sum_{i=1}^n d_i^{(b)}+e_i$. Delete $b$ and $\textbf{e}$.
\item Let $p=\Pr[B_{1/2}=0|S=s]$. Let $B_{p}$ be a Bernoulli variable with\linebreak $\Pr[B_{p}=0]=p$. Choose a realisation $b$ of $B_p$.
\item Choose a realisation $\textbf{x}$ of the random vector $\textbf{X}_2$ with conditional probability 
\[\Pr[\textbf{X}_1=\textbf{x}| B_{1/2}=b, S=s]\]
defined by the random variables from $\mbox{\upshape \sffamily Exp}_{\boldsymbol 1}$.
\item Output $b, \textbf{x}, s$.
\end{enumerate}

First, we show the statistical equivalence of $\mbox{\upshape \sffamily Exp}_1$ and $\mbox{\upshape \sffamily Exp}_2$. Second, we show that $\mbox{\upshape \sffamily Exp}_2$ is efficiently executable.
For showing that $\mbox{\upshape \sffamily Exp}_1$ and $\mbox{\upshape \sffamily Exp}_2$ are statistically equivalent, it suffices to show that the joint distributions $\Pr[B_{1/2},\textbf{X}_1,S]$ and $\Pr[B_{p},\textbf{X}_2,S]$ are equal.

\begin{Lem}[Statistical equivalence of $\mbox{\upshape \sffamily Exp}_{\boldsymbol 1}$ and $\mbox{\upshape \sffamily Exp}_{\boldsymbol 2}$]\label{EXPEQUIV} For the defined experiments, we have $\Pr[B_{1/2},\textbf{\upshape X}_1,S]=\Pr[B_{p},\textbf{\upshape X}_2,S]$.
\end{Lem}
\begin{proof} 
It holds that
\begin{align*}
\Pr[B_{p},\textbf{X}_2,S] & =\Pr[B_p]\cdot\Pr[\textbf{X}_2]\cdot\Pr[S]\\
& = \Pr[B_{1/2}| S]\cdot\Pr[\textbf{X}_1| B_{1/2},S]\cdot\Pr[S]\\
& = \Pr[B_{1/2},\textbf{X}_1,S]. 
\end{align*}
\hfill$\qed$\end{proof}

Note that Lemma \ref{EXPEQUIV} also holds for the marginals of $(B_{1/2},\textbf{X}_1,S)$ and\linebreak $(B_{p},\textbf{X}_2,S)$.\\
For the next lemma, we define $\textbf{E}^{(i)}=(E_i,\ldots,E_n)$ for $i=1,\ldots,n$, where for $i=1,\ldots,n$ the random variable $E_i$ is defined as in $\mbox{\upshape \sffamily Exp}_{\boldsymbol 1}$ (thus $\textbf{E}^{(1)}=\textbf{E}$) and for all $i=1,\ldots,n$, let $\chi^{(i)}:\widehat{\mathcal{D}}^{n-i+1}\to[0,1]$ be the according probability mass function on $\textbf{E}^{(i)}$. Furthermore, let $E^{(i)}=\sum_{j=i}^n E_j$ for all $i=1,\ldots,n$.

\begin{Lem}\label{exp2efficient}
Let $\kappa$ be a complexity parameter. Assume the following holds for all $i=1\ldots,n$:
\begin{enumerate}
\item Given $y\in\mathbb{Z}$, the probabilities $\Pr[E_i=y]$ and $\Pr[E^{(i)}=y]$ are efficiently computable.
\item Either every quantile of $E_i$ is efficiently computable or the number of possible realisations of $E^{(i)}$ is bounded by $t=\text{\upshape\sffamily poly}(\kappa)$ with probability $1-\text{\upshape\sffamily neg}(\kappa)$.
\end{enumerate}
If Experiment $\mbox{\upshape \sffamily Exp}_1$ terminates in time polynomial in $\kappa$, then with probability $1-\text{\upshape\sffamily neg}(\kappa)$, Experiment $\mbox{\upshape \sffamily Exp}_2$ terminates in time polynomial in $\kappa$.
\end{Lem}
\begin{proof} 
Since $\mbox{\upshape \sffamily Exp}_1$ terminates in polynomial time, Step $1$ of $\mbox{\upshape \sffamily Exp}_2$ terminates in polynomial time. Moreover, we have
\begin{align*} p = & \Pr[B_{1/2}=0|S=s]\\
 = & \frac{\Pr[S=s| B_{1/2}=0]\cdot\Pr[B_{1/2}=0]}{\Pr[S=s]}\\
 = & \frac{\Pr[S=s| B_{1/2}=0]}{\Pr[S=s| B_{1/2}=0]+\Pr[S=s| B_{1/2}=1]}\\
 = & \frac{\Pr[E=s-\sum_{i=1}^n d_i^{0}]}{\Pr[E=s-\sum_{i=1}^n d_i^{0}]+\Pr[E=s-\sum_{i=1}^n d_i^{1}]}
\end{align*}
where $E=E^{(1)}=\sum_{i=1}^n E_i$. By the first assumption, all probabilities in the fraction are efficiently computable, thus $p$ is efficiently computable. Since $B_p$ is a Bernoulli variable, Step $2$ of $\mbox{\upshape \sffamily Exp}_2$ terminates in polynomial time. We show that Step $3$ of $\mbox{\upshape \sffamily Exp}_2$ terminates in polynomial time with overwhelming probability. By definition, the random variable $\textbf{X}_2$ is generated by generating $\textbf{X}_1$ conditioned on $B_{1/2}$ and $S$. Let $y=s-\sum_{i=1}^n d_i^{(b)}, \textbf{e}=(e_1,\ldots,e_n)=\textbf{x}-D_b$ and $\textbf{e}^{(i)}=(e_i,\ldots,e_n)$ for all $i=1,\ldots,n$. Let $\mathbbm{1}(y)=\mathbbm{1}_{\{z\in\mathbb{Z} | z=\sum_{i=1}^n e_i\}}(y)$ denote the characteristic function of $\{z\in\mathbb{Z} | z=\sum_{i=1}^n e_i\}$, which is $1$, if $y=\sum_{i=1}^n e_i$ and $0$ otherwise. Then we have
\begin{align*} \Pr[\textbf{X}_2=\textbf{x}] & =\Pr[\textbf{X}_1=\textbf{x}| B_{1/2}=b, S=s]\\
 & =\Pr\left[\textbf{E}=\textbf{x}-D_b\,|\,E=s-\sum_{i=1}^n d_i^{(b)}\right]\\
 & =\Pr[\textbf{E}=\textbf{e}\,|\,E=y]\\
 & =\Pr[\textbf{E}^{(1)}=\textbf{e}^{(1)}\,|\,E^{(1)}=y]\\
 & =\Pr[E_1=e_1\,|\,E^{(1)}=y]\cdot\Pr[\textbf{E}^{(2)}=\textbf{e}^{(2)}\,|\,E_1=e_1, E^{(1)}=y]\\
 & =\Pr[E_1=e_1\,|\,E^{(1)}=y]\cdot\Pr[\textbf{E}^{(2)}=\textbf{e}^{(2)}\,|\,E^{(2)}=y-e_1]\\
 & =\left(\prod_{i=1}^{n-1}\Pr\left[E_i=e_i\,|\,E^{(i)}=y-\sum_{j=1}^{i-1} e_j\right]\right)\cdot\\
 & \mbox{\,\,\,\,\,\,}\cdot\Pr\left[\textbf{E}^{(n)}=\textbf{e}^{(n)}\,|\,E^{(n)}=y-\sum_{j=1}^{n-1} e_j\right]\\
 & =\left(\prod_{i=1}^{n-1}\Pr\left[E_i=e_i\,|\,E^{(i)}=y-\sum_{j=1}^{i-1} e_j\right]\right)\cdot\mathbbm{1}(y).\numberthis\label{determinex_n}
\end{align*}
The second to last equality follows from expanding the recursion. Moreover, for all $i=1\ldots,n-1$ we have
\begin{align*} \Pr\left[E_i=e_i\,|\,E^{(i)}=y-\sum_{j=1}^{i-1} e_j\right] & = \frac{\Pr[E^{(i)}=y-\sum_{j=1}^{i-1} e_j\,|\,E_i=e_i]\cdot\Pr[E_i=e_i]}{\Pr[E^{(i)}=y-\sum_{j=1}^{i-1} e_j]}\\
 & = \frac{\Pr[E^{(i+1)}=y-\sum_{j=1}^{i} e_j]\cdot\Pr[E_i=e_i]}{\Pr[E^{(i)}=y-\sum_{j=1}^{i-1} e_j]}.\numberthis\label{howtocompuee_i}
\end{align*}
Thus, we can implement Step $3$ by the following procedure:
\begin{enumerate}
\item Sequentially, for all $i=1,\ldots,n-1$: choose a realisation $e_i$ of $E_i$ given $E^{(i)}=y-\sum_{j=1}^{i-1} e_j$ according to the distribution given by Equation \ref{howtocompuee_i}.
\item Then the value $E_n=e_n$ is determined by Equation \ref{determinex_n}.
\item Output $\textbf{x}=(d_1^{(b)}+e_1,\ldots,d_n^{(b)}+e_n)$.
\end{enumerate}
By the second assumption, the first step can be performed efficiently either by inverse transform sampling as described in \cite{115} or by considering $(t+1)$-quantiles.\\ 
\hfill$\qed$\end{proof}

\begin{Not} The first assumption of Lemma $\ref{exp2efficient}$ is always satisfied if the generation of $S$ given $B_{1/2}$ is performed by an efficient mechanism that preserves $\epsilon-\mbox{\upshape\sffamily DP}$ in the distributed model, since every user has to perform a perturbation procedure on her own. In that case, also the efficient performance of $\mbox{\upshape \sffamily Exp}_1$ is given. The second assumption of Lemma $\ref{exp2efficient}$ is true for all distributions that are used for the mechanisms considered in this work.
\end{Not}
 
We bound the probability $p$ in $\mbox{\upshape \sffamily Exp}_2$ if the generation of $S$ given $B_{1/2}$ is performed by a mechanism that preserves $\epsilon-\mbox{\upshape\sffamily DP}$.

\begin{Lem}\label{pminmax} Let $B_{1/2}$ be a Bernoulli variable with $\Pr[B_{1/2}=0]=1/2$. Let $\mathcal{A}$ be an $\epsilon-\mbox{\upshape\sffamily DP}$ mechanism and let $D_0,D_1$ be adjacent databases. Let $S=\mathcal{A}(D_{B_{1/2}})$ be a random variable. Then 
\[\frac{1}{1+e^\epsilon}\leq\Pr[B_{1/2}=0| S=s]\leq\frac{e^\epsilon}{1+e^\epsilon}.\]
\end{Lem}
\begin{proof} We bound the probability $p=\Pr[B_{1/2}=0| S=s]$ for the biased coin in $\mbox{\upshape \sffamily Exp}_2$. Since the generation of $S$ was performed by $\mathcal{A}$, we have
\begin{align*} e^{-\epsilon}\cdot\Pr[S=s| B_{1/2}=0] & \leq \Pr[S=s| B_{1/2}=1]\\
& \leq e^\epsilon\cdot\Pr[S=s| B_{1/2}=0].
\end{align*}
By the Bayes rule we get
\[p_{min}:=\frac{1}{1+e^\epsilon}\leq p\leq\frac{e^\epsilon}{1+e^\epsilon}=:p_{max}.\]
\hfill$\qed$\end{proof}

\noindent\textit{The Reduction.} With Lemma \ref{EXPEQUIV} in mind, we can show that playing \textbf{game}~$\boldsymbol 0$ is equivalent to playing \textbf{game}~$\boldsymbol 1_{\boldsymbol B}$.

\begin{Lem}\label{gamezeroPone} Let $\kappa$ be a security parameter. Let $\mathcal{T}_0$ be an adversary in \textbf{\upshape game} $\boldsymbol 0$. 
Let $B_{1/2}$ denote the random variable describing the challenge bit $b$ in \textbf{\upshape game} $\boldsymbol 0$ and let $S$ denote the random variable describing the aggregate of $(x_{i,j^*}^{[B_{1/2}]})_{i\in U}$. Let $p$ be the probability of $B_{1/2}=0$ given the choice of $S$ and let $P=\max\{p,1-p\}$. Then for $a,b=0,1$, under the assumptions of Lemma $\ref{exp2efficient}$, there exists an adversary $\mathcal{T}_{1,B}$ in \textbf{\upshape game}~$\boldsymbol 1_{\boldsymbol B}$ with 
\[\Pr[\mathcal{T}_{1,B}=a, B_p=b]=\Pr[\mathcal{T}_0=a, B_{1/2}=b].\footnote{This also shows that $\mathcal{T}_{1,B}$ has the same success probability as $\mathcal{T}_0$.}\] 
\end{Lem}
\begin{proof} We construct a successful adversary $\mathcal{T}_{1,B}$ in \textbf{game}~$\boldsymbol 1_{\boldsymbol B}$ using $\mathcal{T}_0$ as follows:
\noindent\begin{description}
\item\textbf{Setup.} Receive $\kappa, \mbox{\upshape\sffamily pp},T,s_0$ from the \textbf{game}~$\boldsymbol 1_{\boldsymbol B}$-challenger and send it to $\mathcal{T}_0$. Receive $U=\{i_1,\ldots,i_u\}\subseteq[n]$ from $\mathcal{T}_0$ and send it to the challenger. Forward the obtained response $(s_i)_{i\in[n]\setminus U}$ to $\mathcal{T}_0$.
\item\textbf{Queries.} Forward $\mathcal{T}_0$'s queries $(i,j,d_{i,j})$ with $i\in U, j\in[\lambda], d_{i,j}\in\mathcal{D}$ to the challenger, forward the response $c_{i,j}$ to $\mathcal{T}_0$.
\item\textbf{Challenge.} $\mathcal{T}_0$ chooses $j^*\in\{1,\ldots,\lambda\}$ such that no encryption query in $j^*$ was made and queries two adjacent tuples $(d_{i,j^*}^{[0]})_{i\in U},(d_{i,j^*}^{[1]})_{i\in U}$. Choose a realisation $s$ of $S$ according to $\mbox{\upshape \sffamily Exp}_2$. Set $p=\Pr[B_{1/2}=0| S=s]$ and choose $(x_{i,j^*}^{[0]})_{i\in U}, (x_{i,j^*}^{[1]})_{i\in U}$ with probability 
\[\Pr[\textbf{X}_2=(x_{i,j^*}^{[0]})_{i\in U}], \Pr[\textbf{X}_2=(x_{i,j^*}^{[1]})_{i\in U}]\] 
respectively according to $\mbox{\upshape \sffamily Exp}_2$. Send $p,j^*,(x_{i,j^*}^{[0]},x_{i,j^*}^{[1]})_{i\in U}$ to the challenger. Obtain the response $(c_{i,j^*})_{i\in U}$ and forward it to $\mathcal{T}_0$. 
\item\textbf{Queries.} $\mathcal{T}_0$ can make the same type of queries as before with the restriction that no encryption query in $j^*$ can be made.
\item\textbf{Guess.} $\mathcal{T}_0$ gives a guess about the encrypted database. Output the same guess.
\end{description}
The rules of \textbf{game}~$\boldsymbol 1_{\boldsymbol B}$ are preserved, since $\mathcal{T}_{1,B}$ sends two tuples of the same aggregate $s$ to its challenger. On the other hand, since the ciphertexts generated by the challenger are determined by the challenge bit and the collection $(x_{i,j^*}^{[b]})_{i\in U}$, the rules of \textbf{game}~$\boldsymbol 0$ are preserved by Lemma \ref{EXPEQUIV} (the triple $(b,(x_{i,j^*}^{[b]})_{i\in U},s)$ is chosen according to $\mbox{\upshape \sffamily Exp}_2$, which can be performed efficiently by Lemma \ref{exp2efficient}). Therefore $\mathcal{T}_{1,B}$ perfectly simulates \textbf{game}~$\boldsymbol 0$.
\hfill$\qed$\end{proof}

\subsubsection{Proof of \mbox{\upshape\sffamily CDP}}

We have shown that no ppt adversary can win \textbf{game}~$\boldsymbol 1_{\boldsymbol B}$ if the underlying PSA scheme is secure and that \textbf{game}~$\boldsymbol 1_{\boldsymbol B}$ is equivalent to \textbf{game}~$\boldsymbol 0$. If the perturbation process in \textbf{game}~$\boldsymbol 0$ preserves $\epsilon$-{\sffamily DP}, then the whole construction provides $\epsilon$-\mbox{\upshape\sffamily CDP}, as we show now.


\begin{Thm}[$\mbox{\upshape\sffamily DP}$ and $\mbox{\upshape\sffamily AO}1$ give $\mbox{\upshape\sffamily CDP}$]\label{cdptheorem} Let $\mathcal{A}$ be a randomised mechanism that gets as inpout some database $D=(d_1,\ldots,d_n)\in\mathcal{D}^n$, generates some database $X=X(D)=(x_1(d_1),\ldots,x_n(d_n))=(x_1,\ldots,x_n)$ and outputs $s=\sum_{i=1}^n x_i$, such that $\epsilon$-\mbox{\upshape\sffamily DP} for $D$ is preserved. Let $\Sigma$ be a PSA scheme that gets as input values $x_1,\ldots,x_n$, outputs ciphers $c_1=c_1(x_1),\ldots,c_n=c_n(x_n)$ and $s=\sum_{i=1}^n x_i$ and achieves $\mbox{\upshape\sffamily AO}1$. Then the composition of $\Sigma$ with $\mathcal{A}$ achieves $\mbox{\upshape\sffamily AO}1$ and preserves $\epsilon$-\mbox{\upshape\sffamily CDP}.
\end{Thm}
\begin{proof} Then for all ppt adversaries $\mathcal{T}_{1,B}$ in \textbf{game}~$\boldsymbol 1_{\boldsymbol B}$ and $\mathcal{D}_{\mbox{\scriptsize\upshape\sffamily CDP}}=\mathcal{T}_0$ in \textbf{game}~$\boldsymbol 0$, it holds for $a=0,1$:
\begin{align*}
 \Pr[\mathcal{D}_{\mbox{\scriptsize\upshape\sffamily CDP}}=a, B_{1/2}=1] & = \Pr[\mathcal{T}_{1,B}=a, B_p=1]\numberthis\label{eqcdp1}\\
 & \leq p_{max}\cdot\Pr[\mathcal{T}_{1,B}=a| B_p=1]\\
 & = p_{max}\cdot\Pr[\mathcal{T}_{1,B}=a| B_p=0]+\mbox{\upshape\sffamily neg}(\kappa)\numberthis\label{eqone2}\\
 & \leq \frac{p_{max}}{p_{min}}\cdot\Pr[\mathcal{T}_{1,B}=a, B_p=0]+\mbox{\upshape\sffamily neg}(\kappa)\\
 & = e^\epsilon\cdot\Pr[\mathcal{T}_{1,B}=a, B_p=0]+\mbox{\upshape\sffamily neg}(\kappa)\numberthis\label{eqp}\\
 & = e^\epsilon\,\,\cdot\Pr[\mathcal{D}_{\mbox{\scriptsize\upshape\sffamily CDP}}=a, B_{1/2}=0]+\mbox{\upshape\sffamily neg}(\kappa).\numberthis\label{eqcdp2}
\end{align*}
Equation \eqref{eqone2} holds because of Lemma \ref{gamePoneone}, Equation \eqref{eqp} holds because of Lemma \ref{pminmax} and Equations \eqref{eqcdp1} and \eqref{eqcdp2} hold because of Lemma \ref{gamezeroPone}. 
It follows that
\[\Pr[\mathcal{D}_{\mbox{\scriptsize\upshape\sffamily CDP}}=a| B_{1/2}=1]\leq e^\epsilon\cdot \Pr[\mathcal{D}_{\mbox{\scriptsize\upshape\sffamily CDP}}=a| B_{1/2}=0]+\mbox{\upshape\sffamily neg}(\kappa).\]
\hfill$\qed$\end{proof}

As mentioned in the beginning of this section, we consider a mechanism preserving $\epsilon$-\mbox{\upshape\sffamily DP} given some event $Good$. Therefore, also Theorem \ref{cdptheorem} applies to this mechanism given $Good$. Accordingly, the mechanism \textit{unconditionally} preserves $(\epsilon,\delta)$-\mbox{\upshape\sffamily CDP}, where $\delta$ is a bound on the probability that $Good$ does not happen.

\section{A Weak PRF for \mbox{\upshape\sffamily CDP} based on DLWE}\label{lwepsasec}

We are ready to show how Theorem \ref{PSATHEOREM} contributes to build a prospective post-quantum secure PSA scheme for differentially private data analyses with a relatively high accuracy. Concretely, we can build a secure PSA scheme from the DLWE assumption with errors sampled according to the symmetric Skellam distribution. These errors automatically provide enough noise to preserve \mbox{\upshape\sffamily DP}.

\subsection{The Skellam Mechanism for Differential Privacy}\label{dpmech}

In this section we recall the geometric mechanism from \cite{52} and the binomial mechanism from \cite{14} and introduce the Skellam mechanism. Since these mechanisms make use of a discrete probability distribution, they are well-suited for an execution through a secure PSA scheme, thereby preserving \mbox{\upshape\sffamily CDP} as shown in the last section.

\begin{Def}[Symmetric Skellam Distribution \cite{29}] Let $\mu> 0$. A discrete random variable $X$ is drawn according to the symmetric Skellam distribution with parameter $\mu$ (short: $X\leftarrow\text{\upshape Sk}(\mu)$) if its probability distribution function $\psi_{\mu}\colon\mathbb{Z}\mapsto\mathbb{R}$ is $\psi_{\mu}(k)=e^{-\mu}I_k(\mu)$, where $I_k$ is the modified Bessel function of the first kind (see \cite{28}).
\end{Def}

A random variable $X\leftarrow\text{\upshape Sk}(\mu)$ can be generated as the difference of two Poisson variables with mean $\mu$, (see \cite{29}) and is therefore efficiently samplable. We use the fact that the sum of independent Skellam random variables is a Skellam random variable.

\begin{Lem}[Reproducibility of $\text{\upshape Sk}(\mu)$ \cite{29}]\label{sksum} Let $X\leftarrow\text{\upshape Sk}(\mu_1)$ and $Y\leftarrow\text{\upshape Sk}(\mu_2)$ be independent random variables. Then $Z:=X+Y$ is distributed according to $\text{\upshape Sk}(\mu_1+\mu_2)$.
\end{Lem}

An induction step shows that the sum of $n$ i.i.d. symmetric Skellam variables with variance $\mu$ is a symmetric Skellam variable with variance $n\mu$. The proofs of the following two Theorems are based on standard concentration inequalities and are provided in Section \ref{skellamsec} of the appendix.

\begin{Thm}[Skellam Mechanism]\label{privthm} Let $\epsilon>0$. For every database $D\in\mathcal{D}^n$ and query $f$ with sensitivity $S(f)$ the randomised mechanism $\mathcal{A}(D):=f(D)+Y$ preserves $(\epsilon,\delta)$-\mbox{\upshape\sffamily DP}, if $Y\leftarrow\text{\upshape Sk}(\mu)$ with 
\[\mu=\frac{\log(1/\delta)+\epsilon}{1-\cosh(\epsilon/S(f))+(\epsilon/S(f))\cdot\sinh(\epsilon/S(f))}.\]
\end{Thm}

\begin{Rem}\label{skdprem} The bound on $\mu$ from Theorem \text{\upshape \ref{privthm}} is smaller than $2\cdot (S(f)/\epsilon)^2\cdot(\log(1/\delta)+\epsilon)$, thus for the standard deviation $\sqrt{\mu}$ of $Y\leftarrow\text{\upshape Sk}(\mu)$ it holds that $\sqrt{\mu}=O(S(f)\cdot\sqrt{\log(1/\delta)}/\epsilon)$.
\end{Rem}

Executing this mechanism through a PSA scheme requires the use of the known constant $\gamma$ which denotes the a priori estimate of the lower bound on the fraction of uncompromised users. For this case, we provide the accuracy bound for the Skellam mechanism.

\begin{Thm}[Accuracy of the Skellam Mechanism]\label{errorthm} Let $\epsilon>0, 0<\delta<1, S(f)>0$ and let $0<\gamma<1$ be the a priori estimate of the lower bound on the fraction of uncompromised users in the network. By distributing the execution of a perturbation mechanism as described above and using the parameters from Theorem $\ref{privthm}$, we obtain $(\alpha,\beta)$-accuracy with
\[\alpha=\frac{S(f)}{\epsilon}\cdot\left(\frac{1}{\gamma}\cdot\left(\log\left(\frac{1}{\delta}\right)+\epsilon\right)+\log\left(\frac{2}{\beta}\right)\right).\]
\end{Thm}

\begin{figure*}\centering
\includegraphics[trim={1cm 0 0 0},scale=0.34]{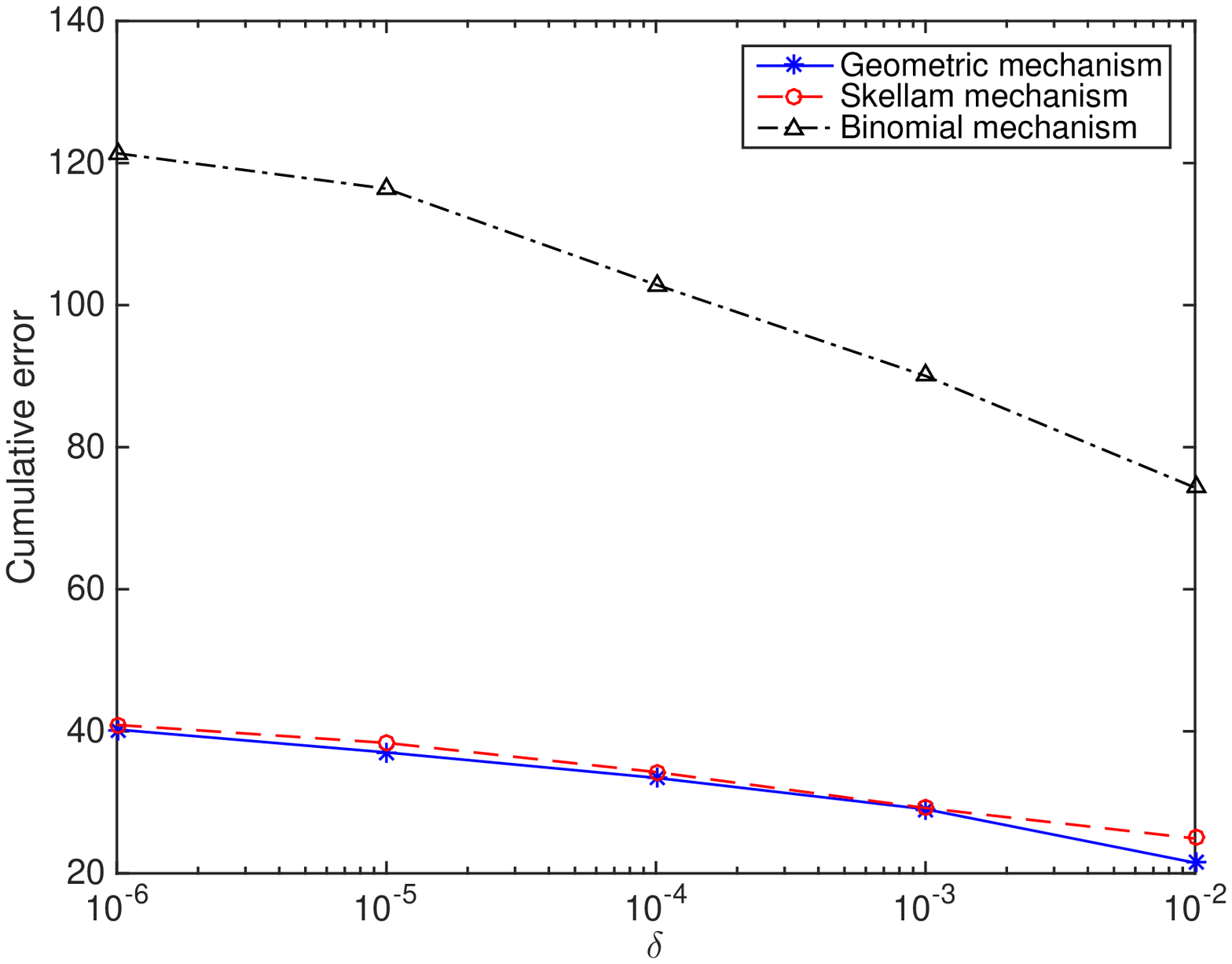}\includegraphics[trim={0.8cm 0 0 0},scale=0.34]{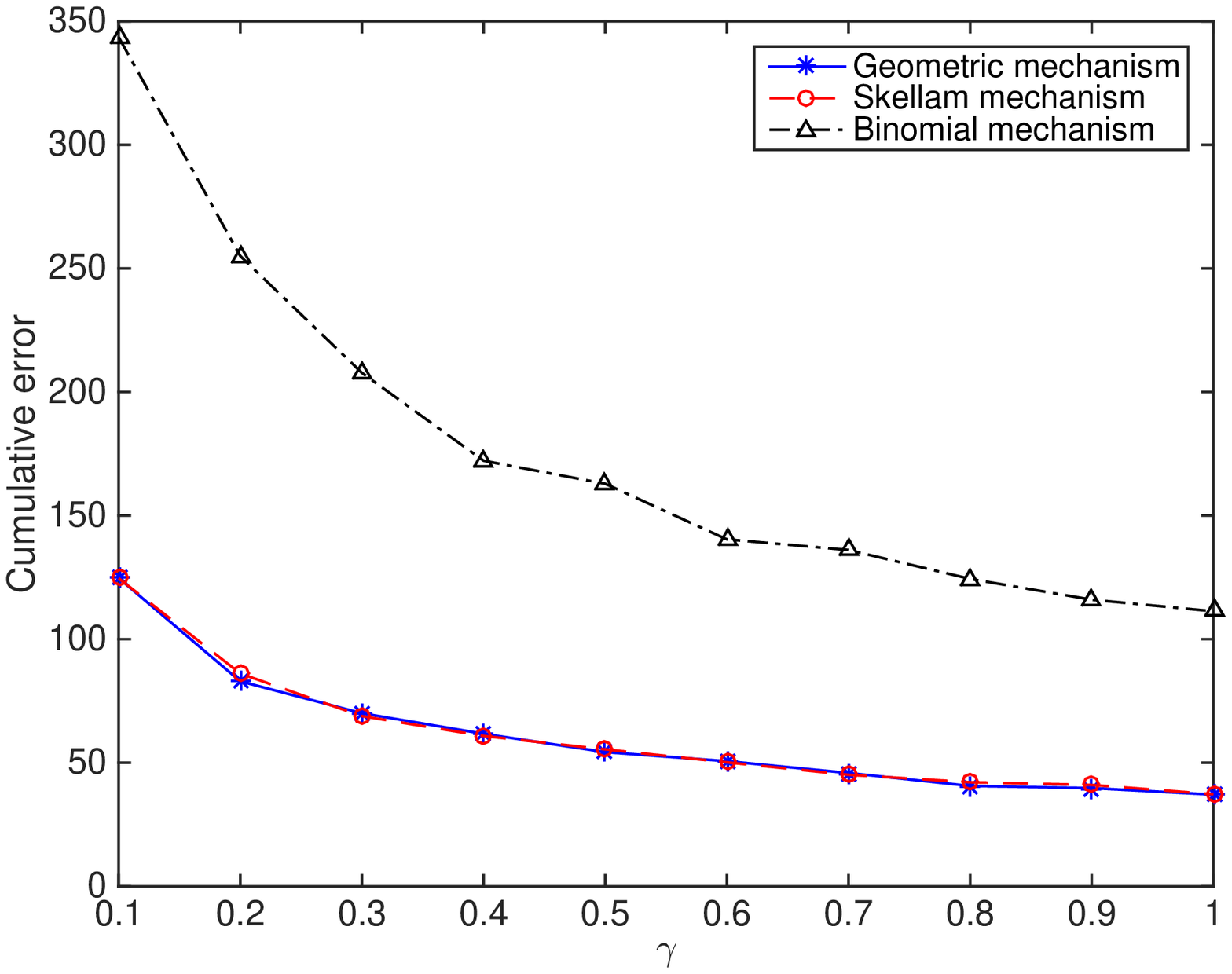}
\caption{Empirical error of the geometric, Skellam and binomial mechanisms. The fixed parameters are $\epsilon=0.1, S(f)=1, n=1000$. The plot on the left shows the mean of the error in absolute value for variable $\delta$ and $\gamma=1$ over $1000$ repeats, the plot on the right is for variable $\gamma$ and $\delta=10^{-5}$.}\label{accuracycomp}
\end{figure*}

Theorem \ref{errorthm} shows that for constant $\delta,\beta,\gamma$ the error of the Skellam mechanism is bounded by $O(S(f)/\epsilon)$. This is the same bound as for the geometric mechanism (see Theorem $3$ in \cite{2}) and the binomial mechanism from \cite{14}. Therefore, the Skellam mechanism has the same accuracy as known solutions. In Figure~\ref{accuracycomp}, an empirical comparison between the mechanisms shows that the error of the geometric and the Skellam mechanisms have a very similar behaviour for both variables $\delta$ and $\gamma$, while the error of the binomial mechanism is roughly three times larger.\\ 
On the other hand, as pointed out in Section \ref{mechov}, the execution of the geometric mechanism through a PSA scheme requires each user to generate full noise with a small probability. Complementary, the Skellam mechanism allows all users to simply generate noise of small variance. This fact makes the Skellam mechanism tremendously advantageous over the geometric mechanism, since it permits to construct a PSA scheme based on the DLWE problem, which automatically preserves \mbox{\upshape\sffamily CDP} without any loss in the accuracy compared to state-of-the-art solutions.

\subsection{Hardness of the LWE problem with Errors following the symmetric Skellam distribution}\label{HARDNESS}

For constructing a secure PSA scheme, we consider the following $\lambda$-bounded (Decisional) Learning with Errors problem and prove the subsequent result.

\begin{Def}[$\boldsymbol\lambda$-bounded LWE] Let $\kappa$ be a security parameter, let $\lambda = \lambda(\kappa) = \text{\upshape\sffamily poly}(\kappa)$ and $q = q(\kappa)\geq 2$ be integers and let $\chi$ be a distribution on $\mathbb{Z}_q$. Let $\textbf{\upshape x}\leftarrow\mathcal{U}(\mathbb{Z}_q^\kappa)$, let $\textbf{\upshape A}\leftarrow\mathcal{U}(\mathbb{Z}_q^{\lambda\times\kappa})$ and let $\textbf{\upshape e}\leftarrow\chi^\lambda$. The goal of the $\text{\upshape LWE}(\kappa,\lambda,q,\chi)$ problem is, given $(\textbf{\upshape A}, \textbf{\upshape Ax} + \textbf{\upshape e})$, to find $\textbf{\upshape x}$. The goal of the $\text{\upshape DLWE}(\kappa,\lambda,q,\chi)$ problem is, given $(\textbf{\upshape A}, \textbf{\upshape y})$, to decide whether $\textbf{\upshape y}=\textbf{\upshape Ax} + \textbf{\upshape e}$ or $\textbf{\upshape y}=\textbf{\upshape u}$ with $\textbf{\upshape u}\leftarrow\mathcal{U}(\mathbb{Z}_q^\lambda)$.
\end{Def}

\begin{Thm}[LWE with Skellam-distributed errors]\label{lweskthm}
Let $\kappa$ be a security parameter and let $\lambda=\lambda(\kappa)=\text{\upshape\sffamily poly}(\kappa)$ with $\lambda>3\kappa$. Let $q=q(\kappa)=\text{\upshape\sffamily poly}(\kappa)$ be a sufficiently large prime modulus and $\rho>0$ such that $\rho q\geq 4\lambda\sqrt{\kappa} s$. If there exists a ppt algorithm that solves the $\text{\upshape LWE}(\kappa,\lambda,q,{\text{\upshape Sk}((\rho q)^2/4)})$ problem with more than negligible probability, then there exists an efficient quantum-algorithm that approximates the decisional shortest vector problem ($\text{\upshape GapSVP}$) and the shortest independent vectors problem ($\text{\upshape SIVP}$) to within $\tilde{O}(\lambda\kappa/\rho)$ in the worst case.
\end{Thm}

Based on the same assumptions, the decisional problem $\text{\upshape DLWE}(\kappa,\lambda,q,$\linebreak ${\text{\upshape Sk}((\rho q)^2/4)})$ is also hard due to the search-to-decision reduction from \cite{43}.\\
Basic notions and facts about the LWE problem can be found in Section \ref{hardnessoflwe} of the appendix. As mentioned in the introduction, our proof uses ideas by D\"{o}ttling and M\"{u}ller-Quade \cite{39}. Similarly to their work, we construct a \textit{lossy code} for the symmetric Skellam distribution from the LWE problem where the errors are taken from the discrete Gaussian distribution $D(\nu)$ with parameter $\nu$. Variants of lossy codes were first used in \cite{53} and since then had applications in different hardness reductions, such as the reduction from the LWE problem to the Learning with Rounding problem from \cite{54}. Lossy codes are pseudo-random codes such that the encoding of a message with the addition of certain errors obliterates any information about the message. On the other hand, encoding the same message using a truly random code and adding the same type of error preserves the message. We will conclude that recovering the message when encoding it with a random code and adding Skellam noise must be computationally hard. If this was not the case, lossy codes could be efficiently distinguished from random codes, contradicting the pseudo-randomness-property of lossy codes.\footnote{Independently of \cite{39}, Bai et al. \cite{84} provide an alternative way to prove the hardness of the LWE problem with error distribution $\chi$. They prove that it is sufficient to show that the R\'{e}nyi divergence between the smoothed distribution $\chi+D(\nu)$ and $\chi$ is sufficiently small (where $D(\nu)$ is the discrete Gaussian with parameter $\nu$, such that the corresponding LWE problem is hard). They realise this proof technique for the uniform error distribution $\chi=\mathcal{U}$. However, the realisation for the Skellam distribution is technically non-trivial.}\\ 
Since the Skellam distribution is both reproducible and well-suited for preserving differential privacy (see Theorem \ref{privthm}), the error terms in our DLWE-based PSA scheme are used for two tasks: establishing the cryptographic security of the scheme and the distributed noise generation to preserve differential privacy.\\
As observed in \cite{39}, considering a $\lambda$-bounded LWE problem, where the adversary is given $\lambda(\kappa) = \text{\upshape\sffamily poly}(\kappa)$ samples, poses no restrictions to most cryptographic applications of the LWE problem, since they require only an a priori fixed number of samples. In our application to differential privacy, we identify $\lambda$ with the number of queries in a pre-defined time-series.\\

\noindent\textit{Entropy and Lossy Codes.} We introduce the conditional min-entropy as starting point for our technical tools. It can be seen as a measure of ambiguity.

\begin{Def}[Conditional min-entropy \cite{39}]
Let $\chi$ be a probability distribution with finite support $Supp(\chi)$ and let $X,\tilde{X}\leftarrow\chi$. Let $f,g$ be two (possibly randomised) maps on the domain $Supp(\chi)$. The $(f,g)$-\text{\upshape conditional min-entropy} $H_{\infty}(X\,|\,f(X)=g(\tilde{X}))$ of $X$ is defined as 
\[H_{\infty}(X\,|\,f(X)=g(\tilde{X}))=-\log_2\left(\max_{\xi\in Supp(\chi)}\{\Pr[X=\xi\,|\,f(X)=g(\tilde{X})]\}\right).\]
\end{Def}

In the remainder of the work we consider $f=f_{\textbf{A},\textbf{e}}$ and $g=g_{\textbf{A},\textbf{e}}$ as maps to the set of LWE instances, i.e.
\[f_{\textbf{A},\textbf{e}}(\textbf{y})=g_{\textbf{A},\textbf{e}}(\textbf{y})=\textbf{Ay} + \textbf{e}.\]
In this work, we consider the $(f_{\textbf{A},\textbf{e}},f_{\textbf{A},\tilde{\textbf{e}}})$-conditional min-entropy
\[H_{\infty}(\textbf{\upshape x}\,|\,f_{\textbf{\upshape A},\textbf{\upshape e}}(\textbf{\upshape x})=f_{\textbf{\upshape A},\tilde{\textbf{\upshape e}}}(\tilde{\textbf{\upshape x}}))=-\log_2\left(\max_{\boldsymbol\xi\in\mathbb{Z}_q^\kappa}\left\{\Pr_{(\textbf{x},\textbf{e})}[\textbf{x}=\boldsymbol\xi\,|\,\textbf{Ax} + \textbf{e}=\textbf{A}\tilde{\textbf{x}} + \tilde{\textbf{e}}]\right\}\right)\]
of a random variable $\textbf{x}$, i.e. the min-entropy of $\textbf{\upshape x}$ given that a LWE instance generated with $(\textbf{\upshape A},\textbf{\upshape x},\textbf{e})$ is equal to another LWE instance generated with $(\textbf{\upshape A},\tilde{\textbf{\upshape x}},\tilde{\textbf{e}})$. Now we provide the notion of lossy codes, which is the main technical tool used in the proof of the hardness result. 

\begin{Def}[Families of Lossy Codes \cite{39}]\label{lossydef}
Let $\kappa$ be a security parameter, let $\lambda = \lambda(\kappa) = \text{\upshape\sffamily poly}(\kappa)$ and let $q=q(\kappa)\geq 2$ be a modulus, $\Delta=\Delta(\kappa)$ and let $\chi$ be a distribution on $\mathbb{Z}_q$. Let $\{\mathcal{C}_{\kappa,\lambda,q}\}$ be a family of distributions, where $\mathcal{C}_{\kappa,\lambda,q}$ is defined on $\mathbb{Z}_q^{\lambda\times\kappa}$. The distribution family $\{\mathcal{C}_{\kappa,\lambda,q}\}$ is $\Delta$-\text{\upshape lossy} for the error distribution $\chi$, if the following hold:
\begin{enumerate}
 \item $\mathcal{C}_{\kappa,\lambda,q}$ is pseudo-random: It holds that $\mathcal{C}_{\kappa,\lambda,q}\approx_c\mathcal{U}(\mathbb{Z}_q^{\lambda\times\kappa})$.
 \item $\mathcal{C}_{\kappa,\lambda,q}$ is lossy: Let $f_{\textbf{\upshape B},\textbf{\upshape b}}(\textbf{\upshape y})=\textbf{\upshape B}\textbf{\upshape y} + \textbf{\upshape b}$. Let $\textbf{\upshape A}\leftarrow\mathcal{C}_{\kappa,\lambda,q}$, $\tilde{\textbf{\upshape x}}\leftarrow\mathcal{U}(\mathbb{Z}_q^\kappa), \tilde{\textbf{\upshape e}}\leftarrow\chi^\lambda$, let $\textbf{\upshape x}\leftarrow\mathcal{U}(\mathbb{Z}_q^\kappa)$ and $\textbf{\upshape e}\leftarrow\chi^\lambda$. Then it holds that 
 \[\Pr_{(\textbf{\upshape A},\tilde{\textbf{\upshape x}},\tilde{\textbf{\upshape e}})}[H_{\infty}(\textbf{\upshape x}\,|\,f_{\textbf{\upshape A},\textbf{\upshape e}}(\textbf{\upshape x})=f_{\textbf{\upshape A},\tilde{\textbf{\upshape e}}}(\tilde{\textbf{\upshape x}}))\geq\Delta]\geq 1-\text{\upshape\sffamily neg}(\kappa).\]
 \item $\mathcal{U}(\mathbb{Z}_q^{\lambda\times\kappa})$ is non-lossy: Let $f_{\textbf{\upshape B},\textbf{\upshape b}}(\textbf{\upshape y})=\textbf{\upshape B}\textbf{\upshape y} + \textbf{\upshape b}$. Let $\textbf{\upshape A}\leftarrow\mathcal{U}(\mathbb{Z}_q^{\lambda\times\kappa})$,\linebreak $\tilde{\textbf{\upshape x}}\leftarrow\mathcal{U}(\mathbb{Z}_q^\kappa), \tilde{\textbf{\upshape e}}\leftarrow\chi^\lambda$, let $\textbf{\upshape x}\leftarrow\mathcal{U}(\mathbb{Z}_q^\kappa)$ and $\textbf{\upshape e}\leftarrow\chi^\lambda$. Then it holds that 
\[\Pr_{(\textbf{\upshape A},\tilde{\textbf{\upshape x}},\tilde{\textbf{\upshape e}})}[H_{\infty}(\textbf{\upshape x}\,|\,f_{\textbf{\upshape A},\textbf{\upshape e}}(\textbf{\upshape x})=f_{\textbf{\upshape A},\tilde{\textbf{\upshape e}}}(\tilde{\textbf{\upshape x}}))=0]\geq 1-\text{\upshape\sffamily neg}(\kappa).\]
\end{enumerate}
\end{Def}

It is not hard to see that the map-conditional entropy suffices for showing that the existence of a lossy code for the error distribution $\chi$ implies the hardness of the LWE problem with error distribution $\chi$.

\begin{Thm}[Lossy code gives hard LWE \cite{39}]\label{lossythm}
Let $\kappa$ be a security parameter, let $\lambda=\lambda(\kappa)=\text{\upshape\sffamily poly}(\kappa)$ and let $q=q(\kappa)$ be a modulus. Let the distribution $\chi$ on $\mathbb{Z}_q$ be efficiently samplable. Let $\Delta=\Delta(\kappa)=\omega(\log(\kappa))$. Then the $\text{\upshape LWE}(\kappa,\lambda,q,\chi)$ problem is hard, given that there exists a family $\{\mathcal{C}_{\kappa,\lambda,q}\}\subseteq\mathbb{Z}_q^{\lambda\times\kappa}$ of $\Delta$-lossy codes for the error distribution $\chi$.
\end{Thm}

Thus, for our purposes it suffices to show the existence of a lossy code for the error distribution $\text{\upshape Sk}(\mu)$. First, it is easy to show that $\mathcal{U}(\mathbb{Z}_q^{\lambda\times\kappa})$ is always non-lossy if the corresponding error distribution $\chi$ can be bounded, thus the third property of Definition \ref{lossydef} is satisfied. 

\begin{Lem}[Non-lossiness of $\mathcal{U}(\mathbb{Z}_q^{\lambda\times\kappa})$ \cite{39}]\label{nonloss}
Let $\kappa$ be a security parameter and $\chi$ a probability distribution on $\mathbb{Z}$. Assume the support of $\chi$ can be bounded by $r=r(\kappa)=\text{\upshape\sffamily poly}(\kappa)$. Moreover, let $q>(4r+1)^{1+\tau}$ for a constant $\tau>0$ and $\lambda=\lambda(\kappa)>(1+2/\tau)\kappa$. Let $f_{\textbf{\upshape B},\textbf{\upshape b}}(\textbf{\upshape y})=\textbf{\upshape B}\textbf{\upshape y} + \textbf{\upshape b}$. Let $\textbf{\upshape A}\leftarrow\mathcal{U}(\mathbb{Z}_q^{\lambda\times\kappa})$, $\tilde{\textbf{\upshape x}}\leftarrow\mathcal{U}(\mathbb{Z}_q^\kappa), \tilde{\textbf{\upshape e}}\leftarrow\chi^\lambda$, let $\textbf{\upshape x}\leftarrow\mathcal{U}(\mathbb{Z}_q^\kappa)$ and $\textbf{\upshape e}\leftarrow\chi^\lambda$. Then
\[\Pr_{(\textbf{\upshape A},\tilde{\textbf{\upshape x}},\tilde{\textbf{\upshape e}})}[H_{f_{\textbf{\upshape A},\textbf{\upshape e}},f_{\textbf{\upshape A},\tilde{\textbf{\upshape e}}},\tilde{\textbf{\upshape x}}}(\textbf{\upshape x})=0]\geq 1-\text{\upshape\sffamily neg}(\kappa).\]
\end{Lem}

For the first and the second properties we construct a lossy code for the Skellam distribution as follows. It is essentially the same construction that was used for the uniform error distribution in \cite{39}.

\begin{Ctn}[Lossy code for the symmetric Skellam distribution]\label{lossysk}
Let $\kappa$ be an even security parameter, let $\lambda=\lambda(\kappa)=\text{\upshape\sffamily poly}(\kappa)$, $\nu>0$ and let $q=q(\kappa)$ be a prime modulus. The distribution $\mathcal{C}_{\kappa,\lambda,q,\nu}$ defined on $\mathbb{Z}_q^{\lambda\times\kappa}$ is specified as follows. Choose $\textbf{\upshape A}^\prime\leftarrow\mathcal{U}(\mathbb{Z}_q^{\lambda\times\kappa/2})$, $\textbf{\upshape T}\leftarrow\mathcal{U}(\mathbb{Z}_q^{\kappa/2\times\kappa/2})$ and $\textbf{\upshape G}\leftarrow D(\nu)^{\lambda\times\kappa/2}$. Output 
\[\textbf{\upshape A}=(\textbf{\upshape A}^\prime||(\textbf{\upshape A}^\prime\textbf{\upshape T}+\textbf{\upshape G})).\]
\end{Ctn}

From the matrix version of the LWE problem and the search-to-decision reduction from \cite{43} (see our Theorem \ref{lwestod}), it is straightforward to see that $\mathcal{C}_{\kappa,\lambda,q,\nu}$ is pseudo-random in the sense of Property $1$ of Definition \ref{lossydef} assuming the hardness of the $\text{\upshape LWE}(\kappa,\lambda,q,D(\nu))$ problem.\\
It remains to show that Construction \ref{lossysk} satisfies Property $2$ of Definition \ref{lossydef}. We first state three supporting claims, whose simple proofs are provided in Section \ref{hardnessoflwe} of the appendix. Let $\textbf{\upshape A}=(\textbf{\upshape A}^\prime||\textbf{\upshape A}^\prime\textbf{\upshape T}+\textbf{\upshape G})$ be the code as defined in Construction \ref{lossysk}. In our further analysis we can consider only $\textbf{\upshape G}$ instead of $\textbf{\upshape A}$.

\begin{Lem}\label{ginsteadofa}
Let $\kappa$ be an even integer, $\textbf{\upshape A}=(\textbf{\upshape A}^\prime||(\textbf{\upshape A}^\prime\textbf{\upshape T}+\textbf{\upshape G}))$ with $\textbf{\upshape A}^\prime\in\mathbb{Z}_q^{\lambda\times\kappa/2}$, $\textbf{\upshape T}\in\mathbb{Z}_q^{\kappa/2\times\kappa/2}$, $\textbf{\upshape G}\in\mathbb{Z}_q^{\lambda\times\kappa/2}$. For all $\textbf{\upshape x}\in\mathbb{Z}_q^{\kappa/2}$ there is a $\textbf{\upshape x}^\prime\in\mathbb{Z}_q^\kappa$ with $\textbf{\upshape Ax}^\prime=\textbf{\upshape Gx}$.
\end{Lem}

\begin{Lem}\label{techlem}
$-C+\sqrt{C^2+1}\geq\exp(-C)$ for all $C\geq 0$.
\end{Lem}

\begin{Lem}\label{smallnormvec}
Let $\kappa$ be a security parameter, let $s=s(\kappa)=\omega(\log(\kappa))$ and let $\nu=\nu(\kappa)=\text{\upshape\sffamily poly}(\kappa)$. Let $\lambda=\lambda(\kappa)=\text{\upshape\sffamily poly}(\kappa), 0<\zeta=\zeta(\kappa)=\text{\upshape\sffamily poly}(\kappa)$ be integers. Let $\textbf{\upshape G}\leftarrow D(\nu)^{\lambda\times\zeta}$. Then for all $\textbf{\upshape z}\in\{0,1\}^\zeta$ the following hold:
\begin{enumerate}
\item $\Pr[||\textbf{\upshape Gz}||_\infty>\zeta\sqrt{\nu}]\leq\text{\upshape\sffamily neg}(\kappa)$, where $||\cdot||_\infty$ is the supremum norm.
\item $\Pr[||\textbf{\upshape Gz}||_2^2>\lambda\zeta^2\nu]\leq\text{\upshape\sffamily neg}(\kappa)$,
where $||\cdot||_2$ is the Euclidean norm.
\end{enumerate}
\end{Lem}

\begin{Lem}\label{maximizing}
Let $\kappa$ be an even security parameter and $\textbf{\upshape A}\in\mathbb{Z}_q^{\lambda\times\kappa}$. Let $s=s(\kappa)=\omega(\log(\kappa))$, let $\mu=\mu(\kappa)$, let $q=\text{\upshape\sffamily poly}(\kappa)$ be a sufficiently large prime modulus and let $\lambda=\lambda(\kappa)=\text{\upshape\sffamily poly}(\kappa)$ be even. Let $\textbf{\upshape e},\tilde{\textbf{\upshape e}}\leftarrow\text{\upshape Sk}(\mu)^\lambda$ and let $\tilde{\boldsymbol\xi}=\argmax_{\boldsymbol\xi\in\mathbb{Z}_q^\kappa}\left\{\Pr_{\textbf{\upshape e}}[\textbf{\upshape e}=\textbf{\upshape A}\boldsymbol\xi+\tilde{\textbf{\upshape e}}]\right\}$. Let $\textbf{\upshape u}=\textbf{\upshape A}\tilde{\boldsymbol\xi}+\tilde{\textbf{\upshape e}}$. Then $||\textbf{\upshape u}||_1\leq\lambda s\sqrt{\mu}$ with probability $1-\text{\upshape\sffamily neg}(\kappa)$.
\end{Lem}

We now show the lossiness of Construction \ref{lossysk} for the error distribution $\text{\upshape Sk}(\mu)$.

\begin{Lem}[Lossiness of Construction $\ref{lossysk}$]\label{lossysklem}
Let $\kappa$ be an even security parameter, $s=s(\kappa)=\omega(\log(\kappa))$, let $\nu=\nu(\kappa)$, let $q=\text{\upshape\sffamily poly}(\kappa)$ be a sufficiently large prime modulus, let $\lambda=\lambda(\kappa)=\text{\upshape\sffamily poly}(\kappa)$ and let $\Delta=\Delta(\kappa)=\omega(\log(\kappa))$. Let $\mu=\mu(\kappa)\geq 4\lambda^2\nu s^2$. Let $f_{\textbf{\upshape B},\textbf{\upshape b}}(\textbf{\upshape y})=\textbf{\upshape B}\textbf{\upshape y} + \textbf{\upshape b}$. Let $\textbf{\upshape A}\leftarrow\{\mathcal{C}_{\kappa,\lambda,q,\nu}\}$ for $\{\mathcal{C}_{\kappa,\lambda,q,\nu}\}$ as in Construction $\text{\upshape\ref{lossysk}}$, $\tilde{\textbf{\upshape x}}\leftarrow\mathcal{U}(\mathbb{Z}_q^\kappa), \tilde{\textbf{\upshape e}}\leftarrow\text{\upshape Sk}(\mu)^\lambda$, $\textbf{\upshape x}\leftarrow\mathcal{U}(\mathbb{Z}_q^\kappa)$ and $\textbf{\upshape e}\leftarrow\text{\upshape Sk}(\mu)^\lambda$. Then
\[\Pr_{(\textbf{\upshape A},\tilde{\textbf{\upshape x}},\tilde{\textbf{\upshape e}})}[H_{\infty}(\textbf{\upshape x}\,|\,f_{\textbf{\upshape A},\textbf{\upshape e}}(\textbf{\upshape x})=f_{\textbf{\upshape A},\tilde{\textbf{\upshape e}}}(\tilde{\textbf{\upshape x}}))\geq\Delta]\geq 1-\text{\upshape\sffamily neg}(\kappa).\]
\end{Lem}
\begin{proof}
Let $(\textbf{\upshape Mz})_j$ denote the $j^{\text{th}}$ entry of $\textbf{\upshape Mz}$ for a matrix $\textbf{\upshape M}$ and a vector $\textbf{\upshape z}$. Let $\textbf{\upshape A}=(\textbf{\upshape A}^\prime||(\textbf{\upshape A}^\prime\textbf{\upshape T}+\textbf{\upshape G}))$ be distributed according to $\mathcal{C}_{\kappa,\lambda,q,\nu}$ with $\textbf{\upshape A}^\prime\leftarrow\mathcal{U}(\mathbb{Z}_q^{\lambda\times\kappa/2})$, $\textbf{\upshape T}\leftarrow\mathcal{U}(\mathbb{Z}_q^{\kappa/2\times\kappa/2})$ and $\textbf{\upshape G}\leftarrow D(\nu)^{\lambda\times\kappa/2}$. Let $\tilde{\textbf{\upshape e}}=(\tilde{e}_j)_{j=1,\ldots,\lambda}\leftarrow\text{\upshape Sk}(\mu)^\lambda$ and let $\tilde{\boldsymbol\xi}=\argmax_{\boldsymbol\xi\in\mathbb{Z}_q^\kappa}\left\{\Pr_{\textbf{\upshape e}}[\textbf{\upshape e}=\textbf{\upshape A}\boldsymbol\xi+\tilde{\textbf{\upshape e}}]\right\}$. Then we have the following chain of (in)equations:
\begin{align*}
 & \Pr_{(\textbf{\upshape A},\tilde{\textbf{\upshape x}},\tilde{\textbf{\upshape e}})}[H_{\infty}(\textbf{\upshape x}\,|\,f_{\textbf{\upshape A},\textbf{\upshape e}}(\textbf{\upshape x})=f_{\textbf{\upshape A},\tilde{\textbf{\upshape e}}}(\tilde{\textbf{\upshape x}}))\geq\Delta]\\
= & \Pr_{(\textbf{\upshape A},\tilde{\textbf{\upshape x}},\tilde{\textbf{\upshape e}})}\left[\max_{\boldsymbol\xi\in\mathbb{Z}_q^\kappa}\left\{\Pr_{(\textbf{x},\textbf{\upshape e})}[\textbf{x}=\boldsymbol\xi\,|\,\textbf{Ax} + \textbf{e}=\textbf{A}\tilde{\textbf{x}} + \tilde{\textbf{e}}]\right\}\leq 2^{-\Delta}\right]\\
= & \Pr_{(\textbf{\upshape A},\tilde{\textbf{\upshape x}},\tilde{\textbf{\upshape e}})}\left[\max_{\boldsymbol\xi\in\mathbb{Z}_q^\kappa}\left\{\Pr_{(\textbf{x},\textbf{\upshape e})}[\textbf{Ax} + \textbf{e}=\textbf{A}\tilde{\textbf{x}} + \tilde{\textbf{e}}\,|\,\textbf{\upshape x}=\boldsymbol\xi]\frac{\Pr_{\textbf{x}}[\textbf{x}=\boldsymbol\xi]}{\Pr_{(\textbf{x},\textbf{e})}[\textbf{\upshape A}\textbf{\upshape x}+\textbf{\upshape e}=\textbf{\upshape A}\tilde{\textbf{x}}+\tilde{\textbf{e}}]}\right\}\right.\\
& \mbox{\,\,\,\,\,\,\,\,\,\,\,\,\,\,\,\,\,\,\,\,\,\,\,\,\,\,\,\,}\left.\leq 2^{-\Delta}\right]\numberthis\label{bayeseq}\\
= & \Pr_{(\textbf{\upshape A},\tilde{\textbf{\upshape x}},\tilde{\textbf{\upshape e}})}\left[\max_{\boldsymbol\xi\in\mathbb{Z}_q^\kappa}\left\{\Pr_{(\textbf{x},\textbf{\upshape e})}[\textbf{Ax} + \textbf{e}=\textbf{A}\tilde{\textbf{x}} + \tilde{\textbf{e}}\,|\,\textbf{\upshape x}=\boldsymbol\xi]\cdot\right.\right.\\ 
& \mbox{\,\,\,\,\,\,\,\,\,\,\,\,\,\,\,\,\,\,\,\,\,\,\,\,\,\,\,\,}\left.\left.\cdot\frac{\Pr_{\textbf{x}}[\textbf{\upshape x}=\boldsymbol\xi]}{\sum_{\textbf{\upshape z}\in\mathbb{Z}_q^\kappa}\Pr_{(\textbf{x},\textbf{\upshape e})}[\textbf{\upshape A}\textbf{\upshape x}+\textbf{\upshape e}=\textbf{\upshape A}\tilde{\textbf{\upshape x}}+\tilde{\textbf{\upshape e}}\,|\,\textbf{\upshape x}=\textbf{\upshape z}]\cdot\Pr_{\textbf{x}}[\textbf{\upshape x}=\textbf{\upshape z}]}\right\}\leq 2^{-\Delta}\right]\\
= & \Pr_{(\textbf{\upshape A},\tilde{\textbf{\upshape x}},\tilde{\textbf{\upshape e}})}\left[\max_{\boldsymbol\xi\in\mathbb{Z}_q^\kappa}\left\{\Pr_{\textbf{\upshape e}}[\textbf{e}=\textbf{A}(\tilde{\textbf{x}}-\boldsymbol\xi) + \tilde{\textbf{e}}]\cdot\right.\right.\\ 
& \mbox{\,\,\,\,\,\,\,\,\,\,\,\,\,\,\,\,\,\,\,\,\,\,\,\,\,\,\,\,}\left.\left.\cdot\frac{\Pr_{\textbf{x}}[\textbf{\upshape x}=\boldsymbol\xi]}{\sum_{\textbf{\upshape z}\in\mathbb{Z}_q^\kappa}\Pr_{\textbf{\upshape e}}[\textbf{\upshape e}=\textbf{\upshape A}(\tilde{\textbf{\upshape x}}-\textbf{z})+\tilde{\textbf{\upshape e}}]\cdot\Pr_{\textbf{x}}[\textbf{\upshape x}=\textbf{\upshape z}]}\right\}\leq 2^{-\Delta}\right]\\
= & \Pr_{(\textbf{\upshape A},\tilde{\textbf{\upshape x}},\tilde{\textbf{\upshape e}})}\left[\max_{\boldsymbol\xi\in\mathbb{Z}_q^\kappa}\left\{\Pr_{\textbf{\upshape e}}[\textbf{e}=\textbf{A}(\tilde{\textbf{x}}-\boldsymbol\xi) + \tilde{\textbf{e}}]\cdot\right.\right.\\ 
& \mbox{\,\,\,\,\,\,\,\,\,\,\,\,\,\,\,\,\,\,\,\,\,\,\,\,\,\,\,\,}\left.\left.\cdot\frac{1}{\sum_{\textbf{\upshape z}\in\mathbb{Z}_q^\kappa}\Pr_{\textbf{\upshape e}}[\textbf{\upshape e}=\textbf{\upshape A}(\tilde{\textbf{\upshape x}}-\textbf{z})+\tilde{\textbf{\upshape e}}]}\right\}\leq 2^{-\Delta}\right]\numberthis\label{applyuniform}\\
= & \Pr_{(\textbf{\upshape A},\tilde{\textbf{\upshape x}},\tilde{\textbf{\upshape e}})}\left[\max_{\boldsymbol\xi\in\mathbb{Z}_q^\kappa}\left\{\Pr_{\textbf{\upshape e}}[\textbf{e}=\textbf{A}\boldsymbol\xi + \tilde{\textbf{e}}]\cdot\frac{1}{\sum_{\textbf{\upshape z}\in\mathbb{Z}_q^\kappa}\Pr_{\textbf{\upshape e}}[\textbf{\upshape e}=\textbf{\upshape A}(\tilde{\textbf{\upshape x}}-\textbf{z})+\tilde{\textbf{\upshape e}}]}\right\}\leq 2^{-\Delta}\right]\numberthis\label{xiorxminusxi}\\
= & \Pr_{(\textbf{\upshape A},\tilde{\textbf{\upshape e}})}\left[\max_{\boldsymbol\xi\in\mathbb{Z}_q^\kappa}\left\{\frac{\Pr_{\textbf{\upshape e}}[\textbf{e}=\textbf{A}\boldsymbol\xi + \tilde{\textbf{e}}]}{\sum_{\textbf{\upshape z}\in\mathbb{Z}_q^\kappa}\Pr_{\textbf{\upshape e}}[\textbf{\upshape e}=\textbf{\upshape A}\textbf{z}+\tilde{\textbf{\upshape e}}]}\right\}\leq 2^{-\Delta}\right]\numberthis\label{eliminatevariable}\\
= & \Pr_{(\textbf{\upshape A},\tilde{\textbf{\upshape e}})}\left[\min_{\boldsymbol\xi\in\mathbb{Z}_q^\kappa}\left\{\frac{\sum_{\textbf{\upshape z}\in\mathbb{Z}_q^\kappa}\Pr_{\textbf{\upshape e}}[\textbf{\upshape e}=\textbf{\upshape A}\textbf{z}+\tilde{\textbf{\upshape e}}]}{\Pr_{\textbf{\upshape e}}[\textbf{e}=\textbf{A}\boldsymbol\xi + \tilde{\textbf{e}}]}\right\}> 2^{\Delta}\right]\\
= & \Pr_{(\textbf{\upshape A},\tilde{\textbf{\upshape e}})}\left[\frac{\sum_{\textbf{\upshape z}\in\mathbb{Z}_q^\kappa}\Pr_{\textbf{\upshape e}}[\textbf{\upshape e}=\textbf{\upshape A}\textbf{z}+\tilde{\textbf{\upshape e}}]}{\Pr_{\textbf{\upshape e}}[\textbf{e}=\textbf{A}\tilde{\boldsymbol\xi} + \tilde{\textbf{e}}]}> 2^{\Delta}\right]\numberthis\label{realisemin}\\
= & \Pr_{(\textbf{\upshape A},\tilde{\textbf{\upshape e}})}\left[\sum_{\textbf{\upshape z}\in\mathbb{Z}_q^\kappa}\frac{\Pr_{\textbf{\upshape e}}[\textbf{\upshape e}=\textbf{\upshape A}(\textbf{z}+\tilde{\boldsymbol\xi})+\tilde{\textbf{\upshape e}}]}{\Pr_{\textbf{\upshape e}}[\textbf{e}=\textbf{A}\tilde{\boldsymbol\xi} + \tilde{\textbf{e}}]}> 2^{\Delta}\right]\numberthis\label{indexshift}\\
= & \Pr_{(\textbf{\upshape A},\tilde{\textbf{\upshape e}})}\left[\sum_{\textbf{\upshape z}\in\mathbb{Z}_q^\kappa}\frac{\prod_{j=1}^\lambda \exp(-\mu)\cdot I_{(\textbf{\upshape A}(\textbf{\upshape z}+\tilde{\boldsymbol\xi}))_j+\tilde{e}_j}(\mu)}{\prod_{j=1}^\lambda \exp(-\mu)\cdot I_{(\textbf{A}\tilde{\boldsymbol\xi})_j+\tilde{e}_j}(\mu)}> 2^{\Delta}\right]\\
= & \Pr_{(\textbf{\upshape A},\tilde{\textbf{\upshape e}})}\left[\sum_{\textbf{\upshape z}\in\mathbb{Z}_q^\kappa}\prod_{j=1}^\lambda\frac{I_{(\textbf{\upshape A}(\textbf{\upshape z}+\tilde{\boldsymbol\xi}))_j+\tilde{e}_j}(\mu)}{I_{(\textbf{A}\tilde{\boldsymbol\xi})_j+\tilde{e}_j}(\mu)}> 2^{\Delta}\right]\\
\geq & \Pr_{(\textbf{\upshape A},\tilde{\textbf{\upshape e}})}\left[\sum_{\textbf{\upshape z}\in\mathbb{Z}_q^\kappa}\prod_{j=1}^\lambda\prod_{k=1+(\textbf{A}\tilde{\boldsymbol\xi})_j+\tilde{e}_j}^{(\textbf{\upshape A}(\textbf{\upshape z}+\tilde{\boldsymbol\xi}))_j+\tilde{e}_j}\frac{-k+\sqrt{k^2+\mu^2}}{\mu}> 2^{\Delta}\right]\numberthis\label{approxmodbes}\\
\geq & \Pr_{(\textbf{\upshape A},\tilde{\textbf{\upshape e}})}\left[\sum_{\textbf{\upshape z}\in\mathbb{Z}_q^\kappa}\prod_{j=1}^\lambda\left(\frac{-((\textbf{\upshape A}(\textbf{\upshape z}+\tilde{\boldsymbol\xi}))_j+\tilde{e}_j)}{\mu}+\sqrt{\left(\frac{(\textbf{\upshape A}(\textbf{\upshape z}+\tilde{\boldsymbol\xi}))_j+\tilde{e}_j}{\mu}\right)^2+1}\right)^{(\textbf{\upshape A}\textbf{\upshape z})_j}\right.\\
&\phantom{\Pr_{(\textbf{\upshape A},\tilde{\textbf{\upshape e}})}\left[\sum_{\textbf{\upshape z}\in\mathbb{Z}_q^\kappa}\prod_{j=1}^\lambda\left(\frac{-((\textbf{\upshape A}\textbf{\upshape z})_j+\tilde{e}_j)}{\mu}\right)\right]}\left.> 2^{\Delta}\right]\numberthis\label{mondeceq}\\
\geq & \Pr_{(\textbf{\upshape A},\tilde{\textbf{\upshape e}})}\left[\sum_{\textbf{\upshape z}\in\mathbb{Z}_q^\kappa}\prod_{j=1}^\lambda\exp\left(-\frac{(\textbf{\upshape A}(\textbf{\upshape z}+\tilde{\boldsymbol\xi}))_j+\tilde{e}_j}{\mu}\right)^{(\textbf{\upshape A}\textbf{\upshape z})_j}> 2^{\Delta}\right]\numberthis\label{techlemeq}\\
= & \Pr_{(\textbf{\upshape A},\tilde{\textbf{\upshape e}})}\left[\sum_{\textbf{\upshape z}\in\mathbb{Z}_q^\kappa}\prod_{j=1}^\lambda\exp\left(-\frac{(\textbf{\upshape A}\textbf{\upshape z})_j^2+(\textbf{\upshape A}\tilde{\boldsymbol\xi})_j\cdot(\textbf{\upshape A}\textbf{\upshape z})_j+\tilde{e}_j\cdot(\textbf{\upshape A}\textbf{\upshape z})_j}{\mu}\right)> 2^{\Delta}\right]\\
= & \Pr_{(\textbf{\upshape A},\tilde{\textbf{\upshape e}})}\left[\sum_{\textbf{\upshape z}\in\mathbb{Z}_q^\kappa}\prod_{j=1}^\lambda\exp\left(-\frac{(\textbf{\upshape A}\textbf{\upshape z})_j^2+(\textbf{\upshape A}\textbf{\upshape z})_j\cdot((\textbf{\upshape A}\tilde{\boldsymbol\xi})_j+\tilde{e}_j)}{\mu}\right)> 2^{\Delta}\right]\\
\geq & \Pr_{(\textbf{\upshape A},\tilde{\textbf{\upshape e}})}\left[\sum_{\textbf{\upshape z}\in\mathbb{Z}_q^\kappa}\exp\left(-\frac{||\textbf{\upshape A}\textbf{\upshape z}||_2^2+||\textbf{\upshape A}\textbf{\upshape z}||_\infty\cdot||\textbf{\upshape A}\tilde{\boldsymbol\xi}+\tilde{\textbf{e}}||_1}{\mu}\right)> 2^{\Delta}\right]\numberthis\label{hoelderineq}\\
\geq & \Pr_{\textbf{\upshape A}}\left[\sum_{\textbf{\upshape z}\in\mathbb{Z}_q^\kappa}\exp\left(-\frac{||\textbf{\upshape A}\textbf{\upshape z}||_2^2+||\textbf{\upshape A}\textbf{\upshape z}||_\infty\cdot\lambda s\sqrt{\mu}}{\mu}\right)> 2^{\Delta}\right]-\text{\upshape\sffamily neg}(\kappa)\numberthis\label{maximizingineq}\\
\geq & \Pr_{\textbf{\upshape G}}\left[\sum_{\textbf{\upshape z}\in\mathbb{Z}_q^{\kappa/2}}\exp\left(-\frac{||\textbf{\upshape G}\textbf{\upshape z}||_2^2+||\textbf{\upshape G}\textbf{\upshape z}||_\infty\cdot\lambda s\sqrt{\mu}}{\mu}\right)> 2^{\Delta}\right]-\text{\upshape\sffamily neg}(\kappa).\numberthis\label{fromAtoG}
\end{align*}

Equation \eqref{bayeseq} is an application of the Bayes rule and Equation \eqref{applyuniform} applies, since $\textbf{\upshape x}$ is sampled according to a uniform distribution. Equation \eqref{xiorxminusxi} applies, since maximising over $\boldsymbol\xi$ is the same as maximising over $\tilde{\textbf{x}}-\boldsymbol\xi$. Equation \eqref{eliminatevariable} is valid since in the denominator we are summing over all possible $\textbf{\upshape z}\in\mathbb{Z}_q^\kappa$. Equation \eqref{realisemin} holds by definition of $\tilde{\boldsymbol\xi}$. Equation \eqref{indexshift} is an index shift by $\tilde{\boldsymbol\xi}$. Inequation \eqref{approxmodbes} follows from essential properties of the modified Bessel functions (iterative application of Lemma \ref{modbesrat}). Note that the modified Bessel function of the first kind is symmetric when considered over integer orders. Therefore, from this point of the chain of (in)equations (i.e. from Inequation \eqref{approxmodbes}), we can assume that $\tilde{e}_j\geq 0$. Moreover, we can assume that $(\textbf{\upshape A}\textbf{\upshape z})_j\geq 0$, since otherwise $I_{(\textbf{\upshape A}\textbf{\upshape z})_j+\tilde{e}_j}(\mu)>I_{-(\textbf{\upshape A}\textbf{\upshape z})_j+\tilde{e}_j}(\mu)$. I.e. if $(\textbf{\upshape A}\textbf{\upshape z})_j< 0$, then we implicitly change the sign of the $j^{\text{th}}$ row in the original matrix $\textbf{\upshape A}$ while considering the particular $\textbf{\upshape z}$. In this way, we are always considering the worst-case scenario for every $\textbf{\upshape z}$. Note that this step does not change the distribution of $\textbf{\upshape A}$, since $\{\mathcal{C}_{\kappa,\lambda,q,\nu}\}$ is symmetric. Inequation \eqref{mondeceq} holds, since $f_\mu(k)=(-k+\sqrt{k^2+\mu^2})/\mu$ is a monotonically decreasing function. Inequation \eqref{techlemeq} follows from Lemma \ref{techlem} by setting $C=((\textbf{\upshape A}\textbf{\upshape z})_j+\tilde{e}_j)/\mu$. Inequation \eqref{hoelderineq} holds because of the H{\"o}lder's inequality.
Inequation \eqref{maximizingineq} follows from Lemma \ref{maximizing}. Inequation \eqref{fromAtoG} follows from Lemma \ref{ginsteadofa}, since 
\[\textbf{\upshape A}=(\textbf{\upshape A}^\prime||(\textbf{\upshape A}^\prime\textbf{\upshape T}+\textbf{\upshape G})).\]
Now consider the set $\mathcal{Z}=\{0,1\}^{\kappa/2}$. Then $|\mathcal{Z}|=2^{\kappa/2}$. Since $\mu\geq 4\lambda^2\nu s^2$, from Lemma \ref{smallnormvec}, it follows that
\begin{align*} & \Pr_{(\textbf{\upshape G},\tilde{\textbf{\upshape e}})}\left[\sum_{\textbf{\upshape z}\in\mathcal{Z}}\exp\left(-\frac{||\textbf{\upshape G}\textbf{\upshape z}||_2^2+||\textbf{\upshape G}\textbf{\upshape z}||_\infty\cdot\lambda s\sqrt{\mu}||_1}{\mu}\right)\geq 2^{\kappa/2}\cdot\exp\left(-\frac{\kappa}{4}-\frac{\kappa}{16s^2}\right)\right]\\
\geq & 1-\text{\upshape\sffamily neg}(\kappa),
\end{align*}
where the norm is computed in the central residue-class representation of the elements in $\mathbb{Z}_q$. Moreover we have
\[2^{\kappa/2}\cdot\exp\left(-\frac{\kappa}{4}-\frac{\kappa}{16s^2}\right)>C^\kappa\]
for some constant $C>1$. Therefore
\begin{align*}
 & \Pr_{(\textbf{\upshape A},\tilde{\textbf{\upshape x}},\tilde{\textbf{\upshape e}})}[H_{\infty}(\textbf{\upshape x}\,|\,f_{\textbf{\upshape A},\textbf{\upshape e}}(\textbf{\upshape x})=f_{\textbf{\upshape A},\tilde{\textbf{\upshape e}}}(\tilde{\textbf{\upshape x}}))\geq\Delta]\\ 
 \geq & \,\Pr_{\textbf{\upshape G}}\left[\sum_{\textbf{\upshape z}\in\mathcal{Z}}\exp\left(-\frac{||\textbf{\upshape G}\textbf{\upshape z}||_2^2+||\textbf{\upshape G}\textbf{\upshape z}||_\infty\cdot\lambda s\sqrt{\mu}}{\mu}\right)> 2^{\Delta}\right]-\text{\upshape\sffamily neg}(\kappa)\\
\geq & \,1-\text{\upshape\sffamily neg}(\kappa)-\text{\upshape\sffamily neg}(\kappa)\\
= & \,1-\text{\upshape\sffamily neg}(\kappa).
\end{align*}
\hfill$\qed$\end{proof}

\noindent Putting the previous results together, we finally show the hardness of the LWE problem with errors drawn from the symmetric Skellam distribution.

\begin{proof}[Proof of Theorem $\ref{lweskthm}$] By Theorem \ref{wtacr}, the $\text{\upshape LWE}(\kappa,\lambda,q,D(\nu))$ problem is hard for $\nu=(\alpha q)^2/(2\pi)>2\kappa/\pi$, if there exists no efficient quantum algorithm approximating the decisional shortest vector problem ($\text{\upshape GapSVP}$) and the shortest independent vectors problem ($\text{\upshape SIVP}$) to within $\tilde{O}(\kappa/\alpha)$ in the worst case. Let $q=q(\kappa)=\text{\upshape\sffamily poly}(\kappa)$, $s=s(\kappa)=\omega(\log(\kappa))$ and $\lambda>3\kappa$. Then for $\Delta=\omega(\log(\kappa))$, Lemma \ref{nonloss}, the pseudo-randomness of Construction \ref{lossysk} and Lemma \ref{lossysklem} provide that Construction \ref{lossysk} gives us a family of $\Delta$-lossy codes for the symmetric Skellam distribution with variance $\mu\geq 4\lambda^2\nu s^2$. As observed in Theorem \ref{lossythm}, this is sufficient for the hardness of the $\text{\upshape LWE}(\kappa,\lambda,q,\text{\upshape Sk}(\mu))$ problem. Setting $\rho=2\alpha\lambda s$ yields $(\rho q)^2>16\lambda^2\kappa s^2$ and the claim follows.
\hfill$\qed$\end{proof}

\noindent By the search-to-decision reduction from \cite{43} we obtain the hardness of the DLWE problem as a corollary.

\subsection{A \mbox{\upshape\sffamily CDP}-Preserving PSA scheme based on DLWE}

\subsubsection{Security of the scheme.}\label{secsubsec}

We can build an instantiation of Theorem \ref{PSATHEOREM} (without correct decryption) based on the $\text{\upshape DLWE}(\kappa,\lambda,q,\chi)$ problem as follows. Set $S=M=\mathbb{Z}_q^\kappa, G=\mathbb{Z}_q$, choose $\textbf{\upshape s}_i\leftarrow \mathcal{U}(\mathbb{Z}_q^\kappa)$ for all $i=1,\ldots,n$ and $\textbf{\upshape s}_0=-\sum_{i=1}^n \textbf{\upshape s}_i$, set $\text{\upshape\sffamily F}_{\textbf{\upshape s}_i}(\textbf{\upshape t})=\langle \textbf{\upshape t},\textbf{\upshape s}_i\rangle +e_i$ (which is a so-called \textit{randomised} weak pseudo-random function as described in \cite{49} and in \cite{44}), where $e_i\leftarrow\chi$ (for the uncompromised users) and let $\varphi$ be the identity function. Therefore
\[\langle \textbf{\upshape t},\textbf{\upshape s}_i\rangle +e_i+d_i=c_{i,\textbf{\upshape t}}\leftarrow\text{\upshape PSAEnc}_{\textbf{\upshape s}_i}(\textbf{\upshape t},d_i)\]
for data value $d_i\in\mathbb{Z}_q$, $i=1,\ldots,n$. The decryption function is defined by 
\[\sum_{i=1}^n d_i+\sum_{i=1}^n e_i=\langle \textbf{\upshape t},\textbf{\upshape s}_0\rangle+\sum_{i=1}^n \text{\upshape\sffamily F}_{\textbf{\upshape s}_i}(\textbf{\upshape t})+d_i=\langle \textbf{\upshape t},\textbf{\upshape s}_0\rangle+\sum_{i=1}^n c_{i,\textbf{\upshape t}}=\text{\upshape PSADec}_{\textbf{\upshape s}_0}(\textbf{\upshape t},c_{1,\textbf{\upshape t}},\ldots,c_{n,\textbf{\upshape t}}).\]
Thus, the decryption is not perfectly correct anymore, but yields a noisy aggregate. Let $\gamma\in(0,1]$ be the a priori known fraction of uncompromised users in the network. Then we can construct the following DLWE-based PSA schemes.

\begin{Exm}\label{discretegaussianexample}
Let $\chi=D(\nu/(\gamma n))$ with parameter $\nu/(\gamma n)=2\kappa/\pi$, then the\linebreak $\text{\upshape DLWE}(\kappa,\lambda,q,\chi)$ problem is hard and the above scheme is secure.
\end{Exm}

\begin{Exm}\label{skellamexample}
Let $\chi=\text{\upshape Sk}(\mu/(\gamma n))$ with variance $\mu/(\gamma n)=4\lambda^2\kappa s^2$, where $\lambda=\lambda(\kappa)=\text{\upshape\sffamily poly}(\kappa)$ with $\lambda>3\kappa$, then the $\text{\upshape DLWE}(\kappa,\lambda,q,\chi)$ problem is hard and the above scheme is secure.
\end{Exm}

\begin{Rem} The original result from $\text{\upshape \cite{42}}$ states that the LWE problem is hard in the set $\mathbb{T}=\mathbb{R}/\mathbb{Z}$ when the noise is distributed according to the \mbox{\em continuous} Gaussian distribution (with a certain bound on the variance) modulo $1$. Although the continuous Gaussian distribution is reproducible as well, it does not seem to fit well for a DLWE-based PSA scheme: For data processing reasons the values would have to be discretised. Therefore the resulting noise would follow a distribution which is not reproducible anymore.\footnote{In \cite{85} it was shown that the sum of $n$ discrete Gaussians each with parameter $\sigma^2$ is statistically close to a discrete Gaussian with parameter $\nu=n\sigma^2$ if $\sigma>\sqrt{n}\eta_\varepsilon(\Lambda)$ for some smoothing parameter $\eta_\varepsilon(\Lambda)$ of the underlying lattice $\Lambda$. However, this approach is less suitable for our purpose if the number of users is large, since the aggregated decryption outcome would have a an error with a variance of order $\nu=\Omega(n^2)$ (in example \ref{skellamexample} the variance is only of order $O(\lambda^2\kappa n)$).} 
\end{Rem}

\subsubsection{Differential privacy of the mechanism.}\label{privsubsec}

The total noise $\sum_{i=1}^n e_i$ in Example \ref{skellamexample} is distributed according to $\text{\upshape Sk}(\mu)$ due to Lemma \ref{sksum}. 
Thus, in contrast to the total noise in Example \ref{discretegaussianexample}, the total noise in Example \ref{skellamexample} preserves the distribution of the single noise and can be used for preserving differential privacy of the correct sum by splitting the task of perturbation among the users.\\
Suppose that adding symmetric Skellam noise with variance $\mu$ preserves $(\epsilon,\delta)$-\mbox{\upshape\sffamily DP}. We define $\mu_{user}=\mu/(\gamma n)$. Since the Skellam distribution is reproducible, the noise addition can be executed in a distributed manner: each (uncompromised) user simply adds (independent) symmetric Skellam noise with variance $\mu_{user}$ to her own value in order to preserve the privacy of the final output.

\subsubsection{Accuracy of the mechanism.}\label{acsubsec}

From Theorem \ref{errorthm} we know that the error of the Skellam mechanism executed in a distributed manner among $\gamma n$ uncompromised users does not exceed $O(S(f)\cdot\log(1/\delta)/\epsilon)$ with high probability. Theorem \ref{PSATHEOREM} indicates that the set $T$ contains all the time-frames where a query can be executed. We identify $|T|=\lambda$, i.e. the number of queries is equal to the number of equations in the instance LWE problem.\footnote{A result from \cite{47} indicates that for an efficient and accurate mechanism this number cannot be substantially larger than $n^2$.} Due to sequential composition\footnote{See for instance Theorem $3$ in \cite{48}.}, in order to preserve $(\epsilon,\delta)$-\mbox{\upshape\sffamily DP} for all $\lambda$ queries together, the executed mechanism must preserve $(\epsilon/\lambda,\delta)$-\mbox{\upshape\sffamily DP} for each query. Therefore the following holds: suppose $\text{\upshape Sk}(\mu^\prime)$-noise is sufficient in order to preserve $(\epsilon,\delta)$-\mbox{\upshape\sffamily DP} for a single query. Then, due to Remark \ref{skdprem}, we must use $\text{\upshape Sk}(\lambda^2\mu^\prime)$-noise in order to preserve $(\epsilon,\delta)$-\mbox{\upshape\sffamily DP} for all $\lambda$ queries. By Theorem \ref{errorthm} the error in each query within $T$ is bounded by $O(\lambda S(f)\cdot\log(1/\delta)/\epsilon)$ which is consistent with the effects of sequential composition.

\subsubsection{Combining Security, Privacy and Accuracy.}

Let $S(f)=\lambda w$ (i.e. $f$ is the statistical analysis over the whole time-period $T$) and at any time $j\in [\lambda]$, let the data of each user come from $\{-w/2,\ldots,w/2\}$. For $\mu=2\cdot (\lambda w/\epsilon)^2\cdot(\log(1/\delta)+\epsilon)$, it follows from the previous discussion and Remark \ref{skdprem} that if every user adds $\text{\upshape Sk}(\mu/(\gamma n))$-noise to her data for every $j=1,\ldots,\lambda$, then this suffices to preserve $(\epsilon,\delta)$-\mbox{\upshape\sffamily DP} for all $\lambda$ sum-queries executed during $T$.\\ 
Furthermore, if for a security parameter $\kappa$ we have that $\mu/(\gamma n)=4\lambda^2\kappa s^2$, then we obtain a secure protocol for sum-queries, where the security is based on prospectively hard lattice problems. As we showed in Section \ref{cdp}, a combination of these two results provides $(\epsilon,\delta)$-\mbox{\upshape\sffamily CDP} for all $\lambda$ sum-queries.

\begin{Cor} Let $\delta>0$. For all $j=1,\ldots,\lambda$, the PSA scheme from Example $\ref{skellamexample}$ preserves $(\epsilon,\delta)$-\mbox{\upshape\sffamily CDP}, where for the largest possible\footnote{Of course it is always possible to decrease $\epsilon$ by simply using a larger magnitude for the noise parameter.} $\epsilon=\epsilon(\kappa)$, it holds that
\[\sqrt{\frac{w^2\cdot\log(1/\delta)}{2\gamma n\kappa s^2}}\leq\epsilon\leq\sqrt{\frac{w^2\cdot\log(1/\delta)}{2\gamma n\kappa s^2}}+\frac{w^2}{2\gamma n\kappa s^2}.\]
\end{Cor}
\begin{proof}
Assume that for $\mu=4\gamma n\lambda^2\kappa s^2$, every uncorrupted user in the network adds $\text{\upshape Sk}(\mu/(\gamma n))$-noise to her data for each of the $\lambda$ queries in order to securely encrypt it using the scheme from Example \ref{skellamexample}. Then there exist $\epsilon, \delta$ such that the decryption output preserves $(\epsilon,\delta)$-\mbox{\upshape\sffamily CDP} for all $\lambda$ queries. In order to calculate $\epsilon$, we set
\[2\cdot (w\lambda/\epsilon)^2\cdot(\log(1/\delta)+\epsilon)= 4\gamma n\lambda^2\kappa s^2.\]
Solving for $\epsilon$, we obtain
\[\epsilon=\frac{w^2+\sqrt{w^4+8w^2\cdot\log(1/\delta)\cdot\gamma n\kappa s^2}}{4\gamma n\kappa s^2}.\]
This expression is smaller than
\[\frac{w^2}{4\gamma n\kappa s^2}+\frac{w^2}{4\gamma n\kappa s^2}+\frac{\sqrt{8 w^2\cdot\log(1/\delta)\cdot\gamma n\kappa s^2}}{4\gamma n\kappa s^2}\]
and larger than
\[\frac{\sqrt{8 w^2\cdot\log(1/\delta)\cdot\gamma n\kappa s^2}}{4\gamma n\kappa s^2}.\]
\end{proof}

This indicates that $\epsilon=\epsilon(\kappa)$ depends on $1/\kappa$. Note that this is consistent with the original definition of \mbox{\upshape\sffamily CDP} from \cite{15}. Thus, in addition to a privacy/accuracy trade-off there is also a security/accuracy trade-off. More specifically, depending on $\kappa$ and $n$ we obtain an upper bound on the minimal $(\alpha,\beta)$-accuracy for every single query executed during~$T$:
\begin{align*}
\alpha & = \frac{\lambda w}{\epsilon}\cdot\left(\frac{1}{\gamma}\cdot\left(\log\left(\frac{1}{\delta}\right)+\epsilon\right)+\log\left(\frac{2}{\beta}\right)\right)\\
& \leq \lambda\cdot\sqrt{\frac{2\gamma n\kappa s^2}{\log(1/\delta)}}\cdot\left(\frac{1}{\gamma}\cdot\log\left(\frac{1}{\delta}\right)+\log\left(\frac{2}{\beta}\right)\right)+\frac{\lambda w}{\gamma}\\
& = O(\lambda s\sqrt{\kappa\cdot n}+\lambda w).
\end{align*}

Finally, we are able to prove our main result, Theorem \ref{mainthm}, which follows from the preceding analyses.

\begin{proof}[of Theorem $\ref{mainthm}$] The claim follows from Theorem \ref{cdptheorem} together with Theorem \ref{PSATHEOREM} (instantiated with the efficient constructions in Example \ref{skellamexample}) and from Theorem \ref{privthm} together with Theorem \ref{errorthm}.
\hfill$\qed$\end{proof}

\section{Conclusions}

In this work we continued a line of research opened by the work of Shi et al. \cite{2}. Using the notion of computational differential privacy, we provided a connection between a secure PSA scheme and a mechanism preserving differential privacy by showing a composition theorem saying that a differentially private mechanism preserves \mbox{\upshape\sffamily CDP} if it is executed through a secure PSA scheme. This closes a security reduction chain from key-homomorphic weak PRFs to \mbox{\upshape\sffamily CDP}, which was initiated in \cite{111}. After introducing the Skellam mechanism for differential privacy we constructed the first prospective post-quantum PSA scheme for analyses of large data amounts from large amounts of individuals under differential privacy. The theoretic basis of the scheme is the DLWE assumption with Skellam noise that is used both for security of the scheme and for preserving computational differential privacy.

\addcontentsline{toc}{chapter}{Literatur}
\bibliography{Literatur}

\appendix

\section{The Skellam mechanism}\label{skellamsec}

\subsection{Preliminaries}

As observed before, the distributed noise generation is feasible with a probability distribution function closed under convolution. For this purpose, we recall the Skellam distribution.

\begin{Def}[Skellam Distribution \cite{29}]\label{skellam} Let $\mu_1$, $\mu_2> 0$. A discrete random variable $X$ is drawn according to the Skellam distribution with parameters $\mu_1,\mu_2$ (short: $X\leftarrow\text{\upshape Sk}(\mu_1,\mu_2)$) if it has the following probability distribution function $\psi_{\mu_1,\mu_2}\colon\mathbb{Z}\mapsto\mathbb{R}$:
\[\psi_{\mu_1,\mu_2}(k)=e^{-(\mu_1+\mu_2)}\left(\frac{\mu_1}{\mu_2}\right)^{k/2}I_k(2\sqrt{\mu_1\mu_2}),\]
where $I_k$ is the modified Bessel function of the first kind (see pages $374$--$378$ in \cite{28}).
\end{Def}

A random variable $X\leftarrow\text{\upshape Sk}(\mu_1,\mu_2)$ has variance $\mu_1+\mu_2$ and can be generated as the difference of two random variables drawn according to the Poisson distribution of mean $\mu_1$ and $\mu_2$, respectively (see \cite{29}). Note that the Skellam distribution is not generally symmetric. However, we mainly consider the particular case $\mu_1=\mu_2=\mu/2$ and refer to this symmetric distribution as $\text{Sk}(\mu) = \text{Sk}(\mu/2,\mu/2)$.\\
Suppose that adding symmetric Skellam noise with variance $\mu$ preserves $(\epsilon,\delta)$-\mbox{\upshape\sffamily DP}. Recall that the network is given an a priori known estimate $\gamma$ of the lower bound on the fraction of uncompromised users. We define $\mu_{user}=\mu/(\gamma n)$ and instruct the users to add symmetric Skellam noise with variance $\mu_{user}$ to their own data. If compromised users will not add noise, the total noise will be still sufficient to preserve $(\epsilon,\delta)$-\mbox{\upshape\sffamily DP} by Lemma \ref{sksum}.\\ 
For our analysis, we use the following bound on the ratio of modified Bessel functions of the first kind.

\begin{Lem}[Bound on $I_k(\mu)$ \cite{27}]\label{modbesrat} For real $k$, let $I_k(\mu)$ be the modified Bessel function of the first kind and order $k$. Then
 \[1>\frac{I_k(\mu)}{I_{k-1}(\mu)}>\frac{-k+\sqrt{k^2+\mu^2}}{\mu},\,k\geq 0.\]
\end{Lem}

Moreover, we will use the Tur\'{a}n-type inequality on the modified Bessel functions from \cite{107}.

\begin{Lem}[Tur\'{a}n-type inequality \cite{107}]\label{turanineq} For $k>-1/2$, let $I_k(\mu)$ be the modified Bessel function of the first kind and order $k$. Then for all $\mu>0$:
\[I_k(\mu)^2>I_{k-1}(\mu)\cdot I_{k+1}(\mu).\]
\end{Lem}

For the privacy analysis of the Skellam mechanism, we need a tail bound on the symmetric Skellam distribution.

\begin{Lem}[Tail bound on the Skellam distribution]\label{SKELLAMBOUND} Let $X\leftarrow\text{\upshape Sk}(\mu)$ and let $\sigma>0$. Then, for all $\tau\geq -\sigma\mu$, 
\[\Pr[X>\sigma\mu + \tau] \leq e^{-\mu\left(1-\sqrt{1+\sigma^2}+\sigma\ln(\sigma + \sqrt{1+\sigma^2})\right)-\tau\ln(\sigma + \sqrt{1+\sigma^2})}.\]
\end{Lem}
\begin{proof}
We use standard techniques from probability theory. Applying Markov's inequality, for any $t>0$,
\begin{align*}\Pr[X>\sigma\mu + \tau] & = \Pr[e^{tX}>e^{t(\sigma\mu + \tau)}]\\
 & \leq \frac{\E[e^{tX}]}{e^{t(\sigma\mu + \tau)}}.
\end{align*}
As shown in \cite{30}, for $X\leftarrow\text{Sk}(\mu)$, the moment generating function of $X$ is 
\[\E[e^{tX}]=e^{-\mu(1-\cosh(t))},\]
where $\cosh(t)= (e^t + e^{-t})/2$. Hence, we have
\[\Pr[X>\sigma\mu + \tau]\leq e^{-\mu(1-\cosh(t)+t\sigma)-t\tau}.\]
Fix $t=\ln(\sigma + \sqrt{1+\sigma^2})$. In order to conclude the proof, we observe that $\cosh(\ln(\sigma + \sqrt{1+\sigma^2}))=\sqrt{1+\sigma^2}$.
\hfill$\qed$\end{proof}

\noindent One can easily verify that, for $\sigma>0$,
\[1-\sqrt{1+\sigma^2}+\sigma\ln(\sigma + \sqrt{1+\sigma^2})> 0.\]

\subsection{Analysis of the Skellam mechanism}

In this section, we prove the bound on the variance $\mu$ of the symmetric Skellam distribution as stated in Theorem \ref{privthm} that is needed in order to preserve $(\epsilon,\delta)$-differential privacy and we compute the error that is thus produced.

\subsubsection{Privacy analysis.}


\begin{proof}[of Theorem $\ref{privthm}$] Let $D_0, D_1\in\mathcal{D}^n$ be adjacent databases with $|f(D_0)-f(D_1)|\leq S(f)$. The largest ratio between $\Pr[\mathcal{A}_{Sk}(D_0)=R]$ and $\Pr[\mathcal{A}_{Sk}(D_1)=R]$ is reached when $k:=R-f(D_0)=R-f(D_1)-S(f)\geq 0$, where $R$ is any possible output of $\mathcal{A}_{Sk}$.
Then, by Lemma \ref{modbesrat}, for all possible outputs $R$ of $\mathcal{A}_{Sk}$:
\begin{align*} \frac{\Pr[\mathcal{A}_{Sk}(D_0)=R]}{\Pr[\mathcal{A}_{Sk}(D_1)=R]} & = \frac{\Pr[Y=k]}{\Pr[Y=k+S(f)]}\\
 & = \prod_{j=1}^{S(f)}\frac{\Pr[Y=k+j-1]}{\Pr[Y=k+j]}\\
 & < \prod_{j=1}^{S(f)}\frac{\mu}{-(k+j)+\sqrt{(k+j)^2+\mu^2}}\\
 & \leq e^\epsilon.\numberthis\label{skmechdp}
\end{align*}
Inequality \eqref{skmechdp} holds if $k\leq\sinh(\epsilon/S(f))\cdot\mu-S(f)$, since it implies $k\leq\sinh(\epsilon/S(f))\cdot\mu-j$ and therefore
\[\frac{\mu}{-(k+j)+\sqrt{(k+j)^2+\mu^2}}\leq e^{\epsilon/S(f)}\] 
for all $j=1,\ldots,S(f)$. Applying Lemma \ref{SKELLAMBOUND} with $\sigma=\sinh(\epsilon/S(f))$ and $\tau=-S(f)$, we get
\[\Pr[Y>\sinh(\epsilon/S(f))\cdot\mu-S(f)]\leq e^{-\mu\cdot(1-\cosh(\epsilon/S(f))+(\epsilon/S(f))\cdot\sinh(\epsilon/S(f)))+\epsilon}\]
and this expression is set to be smaller or equal than $\delta$. This inequality is satisfied if
\[\mu\geq\frac{\log(1/\delta)+\epsilon}{1-\cosh(\epsilon/S(f))+(\epsilon/S(f))\cdot\sinh(\epsilon/S(f))}.\]
\hfill$\qed$\end{proof}

\subsubsection{Accuracy analysis.}

\begin{proof}[of Theorem $\ref{errorthm}$] Let 
\[\mu=\frac{\log(1/\delta)+\epsilon}{1-\cosh(\epsilon/S(f))+(\epsilon/S(f))\cdot\sinh(\epsilon/S(f))}\] 
be the bound on the variance for the Skellam mechanism provided in Theorem \ref{privthm} and let the random variable $Y$ denote the total noise induced by the mechanism. As in the proof of Lemma \ref{SKELLAMBOUND} and since $Y\leftarrow\text{\upshape Sk}(\mu)$, for $\alpha^\prime>0$,
\begin{align*} \Pr[|Y|>\alpha^\prime] & = 2\cdot\Pr[Y>\alpha^\prime]\\
& \leq 2\cdot e^{-\mu\cdot(1-\cosh(\epsilon/S(f)))-(\epsilon/S(f))\cdot\alpha^\prime}
\end{align*}
and this expression is set to be equal to $\beta$. Solving this equality for $\alpha^\prime$ yields
\begin{align*} \alpha^\prime & = \frac{S(f)}{\epsilon}\cdot\left(\left(\cosh\left(\frac{\epsilon}{S(f)}\right)-1\right)\cdot\mu+\log\left(\frac{2}{\beta}\right)\right)\\
 & = \frac{S(f)}{\epsilon}\cdot\left(\left(\log\left(\frac{1}{\delta}\right)+\epsilon\right)\cdot\frac{\cosh(\epsilon/S(f))-1}{1-\cosh(\epsilon/S(f))+(\epsilon/S(f))\cdot\sinh(\epsilon/S(f))}\right.\\
 &\phantom{=}\left.+\log\left(\frac{2}{\beta}\right)\right)\\
 & \leq \frac{S(f)}{\epsilon}\cdot\left(\log\left(\frac{1}{\delta}\right)+\epsilon+\log\left(\frac{2}{\beta}\right)\right)\\
 & = \alpha.
\end{align*}
\hfill$\qed$\end{proof}

\noindent For the distributed noise generation, each single user adds symmetric Skellam noise with variance $\mu_{user}=\mu/(\gamma n)$ to her data. The worst case for accuracy is when all $n$ users add noise, thus the total noise $Y$ is a symmetric Skellam variable with variance $\mu/\gamma$ and the accuracy becomes
\[\alpha=\frac{S(f)}{\epsilon}\cdot\left(\frac{1}{\gamma}\cdot\left(\log\left(\frac{1}{\delta}\right)+\epsilon\right)+\log\left(\frac{2}{\beta}\right)\right),\]
proving Theorem \ref{errorthm}.

\section{Additional details for Section \ref{HARDNESS}}\label{hardnessoflwe}

In this section we will provide more notions and facts about the LWE problem and more details for the proof of Theorem \ref{lweskthm}.

\subsection{Learning with Errors.}

In most of the results about lattice-based cryptography the LWE problem is considered with errors sampled according to a discrete Gaussian distribution.

\begin{Def}[Discrete Gaussian distribution \cite{42}]
Let $q$ be an integer and let $\Phi_s$ denote the normal distribution with variance $s^2/(2\pi)$. Let $\overline{\Psi}_\alpha$ denote the discretised Gaussian distribution with variance $(\alpha q)^2/(2\pi)$, i.e. $\overline{\Psi}_\alpha$ is sampled by taking a sample from $\Phi_{\alpha q}$ and performing a randomised rounding (see \cite{44}). Let $D(\nu)$ be the discretised Gaussian distribution with variance $\nu$, i.e. $D(\nu)=\overline{\Psi}_{\sqrt{2\pi\nu}/q}$.
\end{Def}

We consider a $\lambda$-bounded LWE problem, where the adversary is given $\lambda(\kappa) = \text{\upshape\sffamily poly}(\kappa)$ samples (which we can write conveniently in matrix-form). As observed in \cite{39}, this consideration poses no restrictions to most cryptographic applications of the LWE problem, since they require only an a priori fixed number of samples. In our application to differential privacy (see Section \ref{lwepsasec}) we identify $\lambda$ with the number of queries in a pre-defined time-series.\\

\noindent\textit{Problem} $1$. $\boldsymbol\lambda$\textbf{-bounded LWE Search Problem, Average-Case Version.} Let $\kappa$ be a security parameter, let $\lambda = \lambda(\kappa) = \text{\upshape\sffamily poly}(\kappa)$ and $q = q(\kappa)\geq 2$ be integers and let $\chi$ be a distribution on $\mathbb{Z}_q$. Let $\textbf{x}\leftarrow\mathcal{U}(\mathbb{Z}_q^\kappa)$, let $\textbf{A}\leftarrow\mathcal{U}(\mathbb{Z}_q^{\lambda\times\kappa})$ and let $\textbf{e}\leftarrow\chi^\lambda$. The goal of the $\text{\upshape LWE}(\kappa,\lambda,q,\chi)$ problem is, given $(\textbf{A}, \textbf{Ax} + \textbf{e})$, to find $\textbf{x}$.\\

\noindent\textit{Problem} $2$. $\boldsymbol\lambda$\textbf{-bounded LWE Distinguishing Problem.} Let $\kappa$ be a security parameter, let $\lambda = \lambda(\kappa) = \text{\upshape\sffamily poly}(\kappa)$ and $q = q(\kappa)\geq 2$ be integers and let $\chi$ be a distribution on $\mathbb{Z}_q$. Let $\textbf{x}\leftarrow\mathcal{U}(\mathbb{Z}_q^\kappa)$, let $\textbf{A}\leftarrow\mathcal{U}(\mathbb{Z}_q^{\lambda\times\kappa})$ and let $\textbf{e}\leftarrow\chi^\lambda$. The goal of the $\text{\upshape DLWE}(\kappa,\lambda,q,\chi)$ problem is, given $(\textbf{A}, \textbf{y})$, to decide whether $\textbf{y}=\textbf{Ax} + \textbf{e}$ or $\textbf{y}=\textbf{u}$ with $\textbf{u}\leftarrow\mathcal{U}(\mathbb{Z}_q^\lambda)$.\\

\subsection{Basic Facts}\label{fac}

\subsubsection{Facts about the used Distributions.}

We need to find a tail bound for the sum of discrete Gaussian variables such that the tail probability is negligible.

\begin{Lem}[Bound for the $L_1$-norm of a discrete Gaussian vector]\label{disgausbound}
Let $\kappa$ be a complexity parameter and let $\nu=\nu(\kappa)=\text{\upshape\sffamily poly}(\kappa)$. Let $\zeta=\zeta(\kappa)=\text{\upshape\sffamily poly}(\kappa)$. Let $g_1,\ldots,g_\zeta\leftarrow D(\nu)$ be independent discrete Gaussian variables. Then \[\Pr\left[\sum_{i=1}^\zeta |g_i|>\zeta\sqrt{\nu}\right]\leq\text{\upshape\sffamily neg}(\kappa).\]
\end{Lem}
\begin{proof}
Let $s=s(\kappa)=\omega(\log(\kappa))$ with $s^2=o(\zeta)$. Since $g_1,\ldots,g_\zeta$ are discrete Gaussian variables, we can bound them by a continuous Gaussian variable $X$ with variance~$\nu$:
\[\Pr[|g_i|>\sqrt{s\nu}]\leq\Pr[|X|\geq\sqrt{s\nu}]\leq\exp(-s/2)=\text{\upshape\sffamily neg}(\kappa)\]
for all $i=1,\ldots\zeta$. The random variables $|g_i|$ are distributed according to a half-normal distribution (see \cite{116}) and have mean $\sqrt{2\nu/\pi}$. By the Hoeffding bound we obtain for all $t\geq 0$:
\[\Pr\left[\sum_{i=1}^\zeta |g_i|> t+\zeta\cdot\sqrt{2\nu/\pi}\right]\leq\exp\left(-\frac{2t^2}{\zeta\cdot\nu\cdot s}\right)\]
with probability $1-\text{\upshape\sffamily neg}(\kappa)$. Choosing $t=s\sqrt{\zeta\nu}$ yields the claim.
\hfill$\qed$\end{proof}
Moreover, we need a proper lower bound on the symmetric Skellam distribution that holds with probability exponentially close to $0$.

\begin{Lem}[Bound on the Skellam distribution]\label{skellbound}
Let $\kappa$ be a security parameter, let $s=s(\kappa)=\omega(\log(\kappa))$ and let $\mu=\mu(\kappa)=\text{\upshape\sffamily poly}(\kappa)$ with $\mu>s>0$. Let $X\leftarrow\text{\upshape Sk}(\mu)$. Then \[\Pr[X>s\sqrt{\mu}]\leq\text{\upshape\sffamily neg}(\kappa).\]
\end{Lem}
\begin{proof}
The proof is similar to the proof of Lemma \ref{SKELLAMBOUND}. Applying the Laplace transform and the Markov's inequality we obtain for any $t>0$,
\[\Pr[X>s\sqrt{\mu}]=\Pr[e^{tX}>e^{ts\sqrt{\mu}}]\leq\frac{\E[e^{tX}]}{e^{ts\sqrt{\mu}}}.\]
As for the proof of Lemma \ref{SKELLAMBOUND}, we use the moment generating function of $X\leftarrow\text{Sk}(\mu)$, which is 
\[\E[e^{tX}]=e^{-\mu(1-\cosh(t))},\]
where $\cosh(t)= (e^t + e^{-t})/2$. Hence, we have
\[\Pr[X>s\sqrt{\mu}]\leq e^{-\mu(1-\cosh(t))-ts\sqrt{\mu}}< e^{-s(1-\cosh(t))-ts^{3/2}}.\]
Fix $t=\operatorname{arsinh}(1/\sqrt{s})$. Then
\begin{align*} \Pr[X>s\sqrt{\mu}]< & e^{-s(1-\sqrt{1+1/s})-s^{3/2}\operatorname{arsinh}(1/\sqrt{s})}\\
= & e^{-s}\cdot e^{s\cdot(\sqrt{1+1/s}-\sqrt{s}\operatorname{arsinh}(1/\sqrt{s}))}\\
< & e^{-s}\cdot e^{2/3}\\
= & \text{\upshape\sffamily neg}(\kappa).
\end{align*}
To see the last inequality, observe that the function 
\[f(s)=s\cdot(\sqrt{1+1/s}-\sqrt{s}\operatorname{arsinh}(1/\sqrt{s}))\] 
is monotonically increasing and its limit is $2/3$.
\hfill$\qed$\end{proof}

\subsubsection{Facts about Learning with Errors.}

Regev \cite{42} established worst-to-average case connections between conjecturally hard lattice problems and the\linebreak $\text{\upshape LWE}(\kappa,\lambda,q,D(\nu))$ problem.

\begin{Thm}[Worst-to-Average Case \cite{42}]\label{wtacr} Let $\kappa$ be a security parameter and let $q = q(\kappa)$ be a modulus, let $\alpha=\alpha(\kappa)\in(0,1)$ be such that $\alpha q> 2\sqrt{\kappa}$. If there exists a probabilistic polynomial-time algorithm solving the\linebreak $\text{\upshape LWE}(\kappa,\lambda,q,D((\alpha q)^2/(2\pi)))$ problem with non-negligible probability, 
then there exists an efficient quantum algorithm that approximates the decisional shortest vector problem ($\text{\upshape GAPSVP}$) and the shortest independent vectors problem ($\text{\upshape SIVP}$) to within $\tilde{O}(\kappa/\alpha)$ in the worst case.
\end{Thm}

We use the search-to-decision reduction from \cite{43} basing the hardness of Problem $2$ on the hardness of Problem $1$ which works for any error distribution $\chi$ and is sample preserving.

\begin{Thm}[Search-to-Decision \cite{43}]\label{lwestod} Let $\kappa$ be a security parameter, $q = q(\kappa) = \text{\upshape\sffamily poly}(\kappa)$ a prime modulus and let $\chi$ be any distribution on $\mathbb{Z}_q$. Assume there exists a probabilistic polynomial-time distinguisher that solves the $\text{\upshape DLWE}(\kappa,\lambda,q,\chi)$ problem with non-negligible success-probability, then there exists a probabilistic polynomial-time adversary that solves the  $\text{\upshape LWE}(\kappa,\lambda,q,\chi)$ problem with non-negligible success-probability.
\end{Thm}

Finally, we provide a matrix version of Problem $2$. The hardness of this version can be shown by using a hybrid argument as pointed out in \cite{39}.

\begin{Lem}[Matrix version of LWE]\label{lwematver}
Let $\kappa$ be a security parameter, $\lambda=\lambda(\kappa)= \text{\upshape\sffamily poly}(\kappa)$, $\kappa^\prime=\kappa^\prime(\kappa)=\text{\upshape\sffamily poly}(\kappa)$. Assume that the $\text{\upshape DLWE}(\kappa,\lambda,q,\chi)$ problem is hard. Then $(\textbf{\upshape A},\textbf{\upshape AX}+\textbf{\upshape E})$ is pseudo-random, where $\textbf{\upshape A}\leftarrow\mathcal{U}(\mathbb{Z}_q^{\lambda\times\kappa}), \textbf{\upshape X}\leftarrow\mathcal{U}(\mathbb{Z}_q^{\kappa\times\kappa^\prime})$ and $\textbf{\upshape E}\leftarrow\chi^{\lambda\times\kappa^\prime}$.
\end{Lem}

\subsubsection{Facts about Lossy Codes.}

We will use the fact that the existence of a lossy code (Definition \ref{lossydef}) for an error distribution implies the hardness of the associated decoding problems. This means that solving the LWE problem is hard even though with overwhelming probability the secret is information-theoretically unique. The result was shown in \cite{39}.

\begin{proof}[of Theorem \ref{lossythm}]
Due to the non-lossiness of $\mathcal{U}(\mathbb{Z}_q^{\lambda\times\kappa})$ for $\chi$, instances of $\text{\upshape LWE}(\kappa,\lambda,q,\chi)$ have a unique solution with probability $1-\text{\upshape\sffamily neg}(\kappa)$. Now, let $\mathcal{T}$ be a ppt adversary solving the $\text{\upshape LWE}(\kappa,\lambda,q,\chi)$ problem with more than negligible probability $\sigma$. Using $\mathcal{T}$, we will construct a ppt distinguisher $\mathcal{D}_{\mbox{\scriptsize LWE}}$ distinguishing $\mathcal{U}(\mathbb{Z}_q^{\lambda\times\kappa})$ and $\mathcal{C}_{\kappa,\lambda,q}$ with more than negligible advantage.\\
Let $\textbf{\upshape A}\in\mathbb{Z}^{\lambda\times\kappa}$ be the input of $\mathcal{D}_{\mbox{\scriptsize LWE}}$. It must decide whether $\textbf{\upshape A}\leftarrow\mathcal{U}(\mathbb{Z}_q^{\lambda\times\kappa})$ or $\textbf{\upshape A}\leftarrow\mathcal{C}_{\kappa,\lambda,q}$. Therefore, $\mathcal{D}_{\mbox{\scriptsize LWE}}$ samples $\tilde{\textbf{\upshape x}}\leftarrow\mathcal{U}(\mathbb{Z}^\kappa)$ and $\tilde{\textbf{\upshape e}}\leftarrow\chi^\lambda$. It runs $\mathcal{T}$ on input $(\textbf{\upshape A},\textbf{\upshape A}\tilde{\textbf{\upshape x}}+\tilde{\textbf{\upshape e}})$. Then $\mathcal{T}$ outputs some $\textbf{\upshape x}\in\mathbb{Z}^\kappa$. If $\textbf{\upshape x}=\tilde{\textbf{\upshape x}}$, then $\mathcal{D}_{\mbox{\scriptsize LWE}}$ outputs $1$, otherwise it outputs $0$.\\
If $\textbf{\upshape A}\leftarrow\mathcal{U}(\mathbb{Z}_q^{\lambda\times\kappa})$, then $\tilde{\textbf{\upshape x}}$ is unique and then $\textbf{\upshape x}=\tilde{\textbf{\upshape x}}$ with probability $\sigma$. Therefore 
\[\Pr[\mathcal{D}_{\mbox{\scriptsize LWE}}(\textbf{\upshape A})=1\,|\,\textbf{\upshape A}\leftarrow\mathcal{U}(\mathbb{Z}_q^{\lambda\times\kappa})]=\sigma.\] 
If $\textbf{\upshape A}\leftarrow\mathcal{C}_{\kappa,\lambda,q}$, then $\mathcal{T}$ outputs the Bayes-optimal hypothesis which will be the correct value with probability \[\max_{\boldsymbol\xi\in\mathbb{Z}_q^\kappa}\left\{\Pr[\textbf{x}=\boldsymbol\xi\,|\,\textbf{Ax} + \textbf{e}=\textbf{A}\tilde{\textbf{x}} + \tilde{\textbf{e}}]\right\}=2^{-H_{\infty}(\textbf{\upshape x}\,|\,f_{\textbf{\upshape A},\textbf{\upshape e}}(\textbf{\upshape x})=f_{\textbf{\upshape A},\tilde{\textbf{\upshape e}}}(\tilde{\textbf{\upshape x}}))}\leq 2^{-\Delta},\] 
with $f_{\textbf{\upshape B},\textbf{\upshape b}}(\textbf{\upshape y})=\textbf{\upshape B}\textbf{\upshape y} + \textbf{\upshape b}$. This holds with probability $1-\text{\upshape\sffamily neg}(\kappa)$ over the choice of $(\textbf{\upshape A},\tilde{\textbf{\upshape x}},\tilde{\textbf{\upshape e}})$. This probability is negligible in $\kappa$, since $\Delta=\omega(\log(\kappa))$. Therefore \[\Pr[\mathcal{D}_{\mbox{\scriptsize LWE}}(\textbf{\upshape A})=1\,|\,\textbf{\upshape A}\leftarrow\mathcal{C}_{\kappa,\lambda,q}]=\text{\upshape\sffamily neg}(\kappa)\] 
and in conclusion $\mathcal{D}_{\mbox{\scriptsize LWE}}$ distinguishes $\mathcal{U}(\mathbb{Z}_q^{\lambda\times\kappa})$ and $\mathcal{C}_{\kappa,\lambda,q}$ with probability at least $\sigma-\text{\upshape\sffamily neg}(\kappa)$, which is more than negligible.
\hfill$\qed$\end{proof}

\subsection{Additional Proofs for the Hardness Result}\label{hardn}

We provide the proofs of the supporting Lemma \ref{ginsteadofa}, Lemma \ref{techlem}, Lemma \ref{smallnormvec} and Lemma \ref{maximizing}.\\

\begin{proof}[of Lemma $\ref{ginsteadofa}$]
For all $\textbf{\upshape x}\in\mathbb{Z}_q^{\kappa/2}$, set $\textbf{\upshape x}^\prime=(-(\textbf{\upshape Tx})^{\text{\scriptsize tr}}||\textbf{\upshape x}^{\text{\scriptsize tr}})^{\text{\scriptsize tr}}$. Then 
\begin{align*} \textbf{\upshape Ax}^\prime = & (\textbf{\upshape A}^\prime||(\textbf{\upshape A}^\prime\textbf{\upshape T}+\textbf{\upshape G}))\cdot\textbf{\upshape x}^\prime\\
 = & -\textbf{\upshape A}^\prime\textbf{\upshape Tx}+(\textbf{\upshape A}^\prime\textbf{\upshape T}+\textbf{\upshape G})\cdot\textbf{\upshape x}\\
 = & -\textbf{\upshape A}^\prime\textbf{\upshape Tx}+\textbf{\upshape A}^\prime\textbf{\upshape Tx}+\textbf{\upshape Gx}\\
 = & \textbf{\upshape Gx}.
\end{align*}
\hfill$\qed$\end{proof}

\begin{proof}[of Lemma $\ref{techlem}$]
Let $f(C)=(-C+\sqrt{C^2+1})\exp(C)$. Then $f(C)$ is monotonically increasing and $f(0)=1\cdot(-0+\sqrt{0+1})=1$.
\hfill$\qed$\end{proof}

\begin{proof}[of Lemma $\ref{smallnormvec}$] The claims of Lemma \ref{smallnormvec} follow from Lemma \ref{disgausbound}.
\hfill$\qed$\end{proof}

\begin{proof}[of Lemma $\ref{maximizing}$] By assumption, we have
\begin{equation}\label{maximality1}\Pr_{\textbf{\upshape e}}[\textbf{\upshape e}=\textbf{\upshape u}]=\max_{\boldsymbol\xi\in\mathbb{Z}_q^\kappa}\left\{\Pr_{\textbf{\upshape e}}[\textbf{\upshape e}=\textbf{\upshape A}\boldsymbol\xi+\tilde{\textbf{\upshape e}}]\right\}.
\end{equation}
Let $||\textbf{\upshape u}||_1=C$ and $\tilde{\textbf{\upshape u}}=(\lceil C/\lambda\rceil,\lceil C/\lambda\rceil,\ldots,\lfloor C/\lambda\rfloor,\lfloor C/\lambda\rfloor)^{\text{\scriptsize tr}}$, i.e. $\tilde{\textbf{\upshape u}}$ is maximally balanced with $||\tilde{\textbf{\upshape u}}||_1=C$ (ceiling $C/\lambda$ for the first components of $\tilde{\textbf{\upshape u}}$ and flooring for the last ones). First, we show that under all vectors with $L_1$-norm $C$, the vector $\tilde{\textbf{\upshape u}}$ has the maximal probability weight, i.e. we show that
\begin{equation}\label{maximality2} \Pr_{\textbf{\upshape e}}[\textbf{\upshape e}=\tilde{\textbf{\upshape u}}]=\max_{\textbf{\upshape v}\in\mathbb{Z}_q^\lambda\text{ \scriptsize s.t. } ||\textbf{\upshape v}||_1=C}\left\{\Pr_{\textbf{\upshape e}}[\textbf{\upshape e}=\textbf{\upshape v}]\right\}.
\end{equation}
The difference between the largest and the smallest component in $\tilde{\textbf{\upshape u}}$ is at most $1$ (it is $0$ iff $C$ is a multiple of $\lambda$). Let $\textbf{\upshape w}$ be any less balanced vector than $\tilde{\textbf{\upshape u}}$ with $||\textbf{\upshape w}||_1=C$, i.e. let the difference between the largest and the smallest component in $\textbf{\upshape w}$ be at least $2$. Then there exist indices $i_1,i_2\in[\lambda]$, such that $w_{i_2}-w_{i_1}\geq 2$. We construct a vector $\tilde{\textbf{\upshape w}}$ with the following components: \[\tilde{w}_i=\begin{cases} w_i+1 & \text{ if } i=i_1\\ w_i-1 & \text{ if } i=i_2\\ w_i & \text{ else.}\end{cases}\]
Then, due to Lemma \ref{turanineq}, we have that $\Pr_{\textbf{\upshape e}}[\textbf{\upshape e}=\tilde{\textbf{\upshape w}}]>\Pr_{\textbf{\upshape e}}[\textbf{\upshape e}=\textbf{\upshape w}]$. If we iterate this argument until we consider a maximally balanced vector (i.e. a vector with a difference of at most $1$ between its largest and its smallest component), we obtain 
\[\Pr_{\textbf{\upshape e}}[\textbf{\upshape e}=\tilde{\textbf{\upshape u}}]>\Pr_{\textbf{\upshape e}}[\textbf{\upshape e}=\textbf{\upshape w}].\]
This implies Equation \eqref{maximality2}.\\
Now assume that $C>\lambda s\sqrt{\mu}$. Then, by the previous claim, we have that $\tilde{u}_j\geq\lfloor s\sqrt{\mu}\rfloor$ for each component $\tilde{u}_j$ of $\tilde{\textbf{\upshape u}}$, $j=1,\ldots,\lambda$. Since by Lemma \ref{skellbound}, a Skellam variable is bounded by $\lfloor s\sqrt{\mu}\rfloor$ with overwhelming probability, we have 
\[\Pr_{\textbf{\upshape e}}[\textbf{\upshape e}=\tilde{\textbf{\upshape e}}]>\Pr_{\textbf{\upshape e}}[\textbf{\upshape e}=\tilde{\textbf{\upshape u}}]\stackrel{\eqref{maximality2}}{\geq}\Pr_{\textbf{\upshape e}}[\textbf{\upshape e}=\textbf{\upshape u}]\] 
with overwhelming probability over the choice of $\tilde{\textbf{\upshape e}}$, which is a contradiction to Equation \eqref{maximality1}.
\hfill$\qed$\end{proof}

\end{document}